\documentclass[onecolumn,11pt,tightenlines,superscriptaddress]{revtex4-2}

\usepackage{latexsym}
\usepackage{amssymb}
\usepackage{amsmath}
\usepackage{amsthm}
\usepackage{amsfonts}
\usepackage{complexity,physics}
\usepackage{ifthen}
\usepackage{revsymb}
\usepackage{yfonts}
\usepackage{graphicx}
\usepackage{algpseudocode}
\usepackage{bbm}
\usepackage{multirow}
\usepackage{epstopdf}
\usepackage{color}
\usepackage[dvipsnames]{xcolor}
\usepackage{afterpage}
\usepackage{verbatim}
\usepackage{soul}
\usepackage{lineno, blindtext}
\usepackage{footnote}
\usepackage{afterpage}
\usepackage[normalem]{ulem}
\usepackage{algorithm}

\usepackage[colorlinks = true]{hyperref}
\usepackage{longtable}

\usepackage{xcolor}
\definecolor{darkred}  {rgb}{0.5,0,0}
\definecolor{darkblue} {rgb}{0,0,0.5}
\definecolor{darkgreen}{rgb}{0,0.5,0}

\hypersetup{
  urlcolor   = blue,         
  linkcolor  = darkblue,     
  citecolor  = darkgreen,    
  filecolor   = darkred       
}

\newcommand{\calH}{{\cal H }}

\newcommand{\calX}{{\cal X }}

\newcommand{\mc}{\mathcal}
\newcommand{\ints}{\mathbb{Z}}
\newcommand{\Cc}{\mathcal{C}}
\newcommand{\fat}{\mathsf{fat}}

\newcommand{\be}{\begin{equation}}
\newcommand{\ee}{\end{equation}}
\newcommand{\bq}{\begin{eqnarray}}
\newcommand{\eq}{\end{eqnarray}}
\newcommand{\bea}{\begin{eqnarray}}
\newcommand{\eea}{\end{eqnarray}}
\newcommand{\ba}{\begin{align}}
\newcommand{\ea}{\end{align}}

\newcommand{\DLP}{\mathsf{DLP}}
\newcommand{\QKE}{\mathsf{QKE}}
\newcommand{\SVM}{\mathsf{SVM}}
\DeclareMathOperator{\Var}{Var}

\newtheorem{theorem}{Theorem}
\newtheorem{remark}{Remark}
\newtheorem{lemma}[theorem]{Lemma}
\newtheorem{corollary}[theorem]{Corollary}

\newtheorem{definition}[theorem]{Definition}

\newcommand{\pmset}[1]{\{-1,1\}^{#1}}

\renewcommand{\E}{\mathop{\mathbb E\/}}

\definecolor{myexpcolor}{RGB}{139,0,139}
\definecolor{myinnercolor}{RGB}{102,102,0}
\definecolor{myburg}{rgb}{0.7,0.11,0.11}
\definecolor{mygray}{gray}{0.6}

\newenvironment{nalign}{
    \begin{equation}
    \begin{aligned}
}{
    \end{aligned}
    \end{equation}
    \ignorespacesafterend{}
}

\DeclareMathOperator{\sign}{sign}

\begin{document}

\title{A rigorous and robust quantum speed-up in supervised machine learning}

\author{Yunchao Liu}
 \email{yunchaoliu@berkeley.edu}
\affiliation{Department of Electrical Engineering and Computer Sciences, University of California, Berkeley, CA 94720}
\affiliation{IBM Quantum, T.J.\ Watson Research Center, Yorktown Heights, NY 10598}
\author{Srinivasan Arunachalam}
\email{Srinivasan.Arunachalam@ibm.com}
\affiliation{IBM Quantum, T.J.\ Watson Research Center, Yorktown Heights, NY 10598}
\author{Kristan Temme}
\email{kptemme@ibm.com}
\affiliation{IBM Quantum, T.J.\ Watson Research Center, Yorktown Heights, NY 10598}
\date{\today}

\begin{abstract}

Over the past few years several quantum machine learning algorithms were proposed that promise quantum  speed-ups  over  their  classical  counterparts. Most  of  these  learning  algorithms either assume quantum access to data -- making it unclear if quantum speed-ups still exist without making these strong assumptions, or are heuristic in nature with no provable advantage over classical algorithms. In this paper, we establish a rigorous quantum speed-up for supervised classification using a general-purpose quantum learning algorithm that only requires classical access to data. Our quantum classifier is a conventional support vector machine that uses a fault-tolerant quantum computer to estimate a kernel function. Data samples are mapped to a quantum feature space and the kernel entries can be estimated as the transition amplitude of a quantum circuit. We construct a family of datasets and show that no classical learner can classify the data inverse-polynomially better than random guessing, assuming the widely-believed hardness of the discrete logarithm problem. Meanwhile, the quantum classifier achieves high accuracy and is robust against  additive errors in the kernel entries that arise from finite sampling statistics.
\end{abstract}

\maketitle

  Finding potential applications for quantum computing which demonstrate quantum speed-ups is a central goal of the field. Much attention has been drawn towards establishing a quantum advantage in machine learning due to its wide applicability~\cite{Biamonte2017,Arunachalam2017survey,dunjko2018machine,Ciliberto2018quantum,RevModPhys.91.045002}. 
  In this direction there have been several quantum algorithms for machine learning tasks that promise polynomial and exponential speed-ups. A family of such quantum algorithms assumes that classical data is encoded in amplitudes of a quantum state, which uses a number of qubits that is only logarithmic in the size of the dataset. These quantum algorithms are therefore able to achieve exponential speed-ups over classical  approaches~\cite{Harrow2009quantum,wiebe2012quantum,lloyd2013quantum,Lloyd2014,Rebentrost2014quantum,lloyd2014quantumtda,Cong_2016,kerenidis2016quantum,Brandao2019sdp,Rebentrost2018svd,Zhao2019quantum}. However, it is not known whether data can be efficiently provided this way in practically relevant settings. This raises the question of whether the advantage comes from the quantum algorithm, or from the way data is provided~\cite{aaronson2015read}. Indeed, recent works have shown that if classical algorithms have an analogous sampling access to data, then some of the proposed exponential speed-ups do no longer exist~\cite{Tang2019quantuminspired,tang2018quantuminspired,gilyn2018quantuminspired,chia2018quantuminspired,ding2019quantuminspired,Chia2020samplingbased}. 

  Consequently a different class of quantum algorithms has been developed which only assumes classical access to data. Most of these algorithms use variational circuits for learning, where a candidate circuit is selected from a parameterized circuit family via classical optimization~\cite{mitarai2018quantum,farhi2018classification,grant2018hierarchical,schuld2020circuit,Benedetti_2019}. 
  Although friendly to experimental implementation, these algorithms are heuristic in nature since no formal evidence has been provided which shows that they have a genuine advantage over classical algorithms. An important challenge is therefore to find one example of such a heuristic quantum machine learning algorithm, which given \emph{classical access} to data can \emph{provably} outperform all classical learners for some learning~problem.

In this paper, we answer this in the affirmative. We show that an exponential quantum speed-up can be obtained via the use of a quantum-enhanced feature space~\cite{Havlicek2019,Schuld2019quantum}, where each data point is mapped non-linearly to a quantum state and then classified by a linear classifier in the high dimensional Hilbert space. To efficiently learn a linear classifier in feature space from training data, we use the standard kernel method in \emph{support vector machines} ($\SVM$s), a well-known family of supervised classification algorithms~\cite{boser1992training,vapnik2013nature}. We obtain the kernel matrix by measuring the pairwise inner product of the feature states on a quantum computer, a procedure we refer to as \emph{quantum kernel estimation} ($\QKE$). This kernel matrix is then given to a classical optimizer that efficiently finds the linear classifier that optimally separates the training data in feature space by running a convex quadratic program.

The advantage of our quantum learner stems from its ability to recognize classically intractable complex patterns using the feature map. We prove an end-to-end quantum advantage based on this intuition, where our quantum classifier is guaranteed to achieve high accuracy for a classically hard classification problem. We show that under a suitable quantum feature map, the classical data points, which are indistinguishable from having random labels by efficient classical algorithms, are linearly separable with a large margin in high-dimensional Hilbert space. Based on this property, we then combine ideas from classic results on the generalization of soft margin classifiers~\cite{anthony_bartlett_2000,Shawe-Taylor1998structural,BST1999,Shawe-Taylor2002generalization} to rigorously bound the misclassification error of the $\SVM$-$\QKE$ algorithm. The optimization for large margin classifiers in the $\SVM$ program is crucial in our proof, as it allows us to learn the optimal separator in the exponentially large feature space, while also making our quantum classifier robust against additive sampling errors.

Our classification problem that shows the exponential quantum speed-up is constructed based on the discrete logarithm problem ($\DLP$). We prove that no efficient classical algorithm can achieve an accuracy that is inverse-polynomially better than random guessing, assuming the widely-believed classical hardness of $\DLP$. In computational learning theory, the use of one-way functions for constructing classically hard learning problems is a well-known technique~\cite{kearns1990computational}. Rigorous separations between quantum and classical learnability have been established using this idea in the quantum oracular and PAC model~\cite{Servedio2004equivalence,Arunachalam2017survey}, as well as in the classical generative setting~\cite{sweke2020quantum}. There the quantum algorithms are constructed specifically to solve the problems for showing a separation, and in general are not applicable to other learning problems. Based on different complexity-theoretic assumptions, evidences of an exponential quantum speed-up were shown for a quantum generative model~\cite{Gao2018quantum}, where the overall performance is not guaranteed.

\bigskip
{\noindent\bf A classically intractable learning problem}\\[1mm]
The task of supervised classification is to assign a label $y\in\{-1,1\}$ to a datum $x \in \calX$ from data space $\calX$ according to some unknown decision rule $f$ (usually referred to as a \emph{concept}), by learning from labeled examples $S=\{(x_i,y_i)\}_{i=1,\ldots, m}$ that are generated from this concept, $y_i=f(x_i)$. Given the training set $S$, an efficient learner needs to compute a classifier $f^*$ in time that is polynomial in the size of $S$,  
with the goal of achieving high \emph{test accuracy},
\begin{equation}
    \mathrm{acc}_{f}(f^*)=\Pr_{x \sim \calX}\left[f^*(x)=f(x)\right],
\end{equation}
the probability of agreeing with $f$ on unseen examples. Here we assume that the datum $x$ is sampled uniformly random from $\calX$, in both training and testing, and the size of $S$ is polynomial in the data dimension.

An important ingredient of machine learning is prior knowledge, i.e., additional information given to the learning algorithm besides the training set. In standard computational learning theory~\cite{kearns1990computational,kearns1994introduction}, this is modeled as a concept class -- a (often exponentially large) set of labeling rules, and the target concept is promised to be chosen from the concept class. A concept class $\mc C$ is efficiently learnable if for every $f\in \mc C$, an efficient algorithm can achieve 0.99 test accuracy by learning from examples labeled according to $f$ with high success probability. See Appendix~\ref{app:definitions} for detailed definitions. 

Our concept class that separates quantum and classical learnability is based on the discrete logarithm problem ($\DLP$). For a large prime number $p$ and a generator $g$ of $\ints_p^*=\{1,2,\dots,p-1\}$, it is a widely-believed conjecture that no classical algorithm can compute $\log_g x$ on input $x \in \ints_p^*$, in time polynomial in $n = \lceil \log_2(p) \rceil$, the number of bits needed to represent $p$. Meanwhile, $\DLP$ can be solved by Shor's quantum algorithm~\cite{Shor1997polynomial} in polynomial time. 

Based on $\DLP$, we define our concept class $\mc C=\{f_s\}_{s\in \ints_p^*}$ over the data space $\calX=\ints_p^*$ as follows,
\begin{equation}\label{eq:defconceptclass}
    f_s(x)=\begin{cases}+1,&\text{if }\log_g x\in[s,s+\frac{p-3}{2}],\\ -1, &\text{else.}\end{cases}
\end{equation}
Each concept $f_s:\ints_p^*\to \{-1,1\}$ maps half the elements in $\ints_p^*$ to $+1$ and half of them to $-1$.

\begin{figure}[t]
    \centering
    \includegraphics[width=\linewidth]{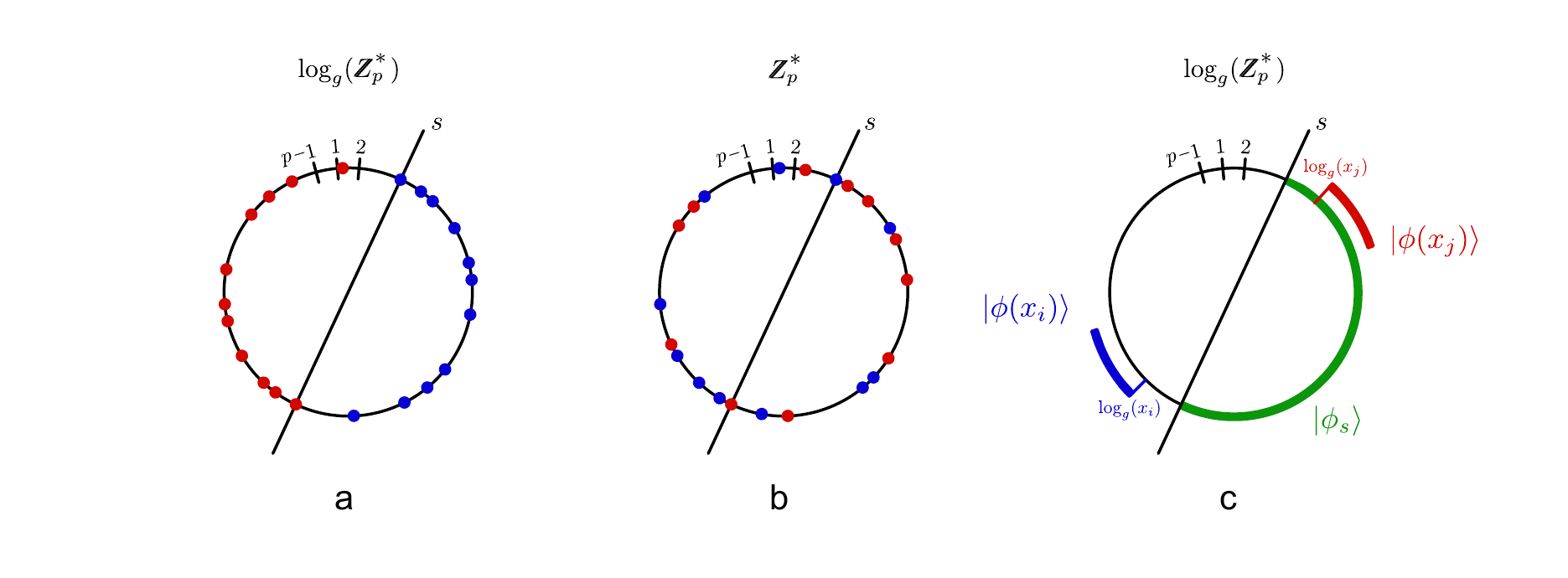}
    \caption{Learning the concept class $\mc C$ by a quantum feature map. (\textbf{a}) After taking the discrete log of the data samples, they become separated in log space by the concept $s$. (\textbf{b}) However, in the original data space, the data samples look like randomly labeled and cannot be learned by an efficient classical algorithm. (\textbf{c}) Using the quantum feature map, each $x\in\ints_p^*$ is mapped to a quantum state $\ket{\phi(x)}$, which corresponds to a uniform superposition of an interval in log space starting with $\log_g x$. This feature map creates a large margin, as the $+1$ labeled example (red interval) has high overlap with a separating hyperplane (green interval), while the $-1$ labeled example (blue interval) has zero overlap.}
    \label{fig:data}
\end{figure}

To see why the discrete logarithm is important in our definition, note that if we change $\log_g x$ to $x$ in Eq.~\eqref{eq:defconceptclass}, then learning the concept class $\mc C$ is a trivial problem. Indeed, if we imagine the elements of $\ints_p^*$ as lying on a circle, then each concept $f_s$ corresponds to a direction for cutting the circle in two halves (Fig.~\ref{fig:data}a). Therefore, the training set of labeled examples can be separated as two distinct clusters, where one cluster is labeled $+1$ and the other labeled $-1$. To classify a new example, a learning algorithm can simply decide based on which cluster is closer to the example. This intuition also explains why the original concept class $\mc C$ is learnable by a quantum learner, since the learner can use Shor's algorithm to compute the discrete log for every data sample it receives~\cite{Servedio2004equivalence}, and then solve the above trivial learning problem.

On the other hand, due to the classical intractability of $\DLP$, the training samples look like randomly labeled from the viewpoint of a classical learner (Fig.~\ref{fig:data}b). In fact, we can prove that the best a classical learner can do is randomly guess the label for new examples, which achieves 50\% test accuracy. These results are summarized below.

\begin{theorem}\label{thm:mainthm1}
Assuming the classical hardness of $\DLP$, no efficient classical algorithm can achieve $\frac{1}{2}+\frac{1}{\poly(n)}$ test accuracy for $\mc C$.
\end{theorem}

Our proof, c.f.\ Appendix~\ref{app:classicalhardness}, of classical hardness of learning $\mc C$ is based on an average-case hardness result for discrete log by Blum and Micali~\cite{blum1984how}. They showed that computing the most significant bit of $\log_g x$ for $\frac{1}{2}+\frac{1}{\poly(n)}$ fraction of $x\in\ints_p^*$ is as hard as solving $\DLP$. We then reduce our concept class learning problem to $\DLP$ using this result, by showing that if an efficient learner can achieve $\frac{1}{2}+\frac{1}{\poly(n)}$ test accuracy for $\mc C$, then it can be used to construct an efficient classical algorithm for~$\DLP$, which proves Theorem~\ref{thm:mainthm1}. 

In addition to establishing a separation between quantum and classical learnability for binary classification, we also note that this separation can be efficiently verified by a classical verifier in an interactive setting. This follows from a nice property of our concept class. We show that for every concept $f\in\mc C$, we can \emph{efficiently generate} labeled examples $(x,f(x))$ classically where $x\sim \ints_p^*$ is uniformly distributed, despite $f(x)$ being hard to compute by definition. To test if a prover can learn $\mc C$, a classical verifier can pick a random concept and efficiently generate two sets of data $(S,T)$, where $S$ is a training set of labeled examples and $T \subseteq \ints_p^*$ is a test set of examples with labels removed. The verifier then sends $(S,T)$ to the prover and asks for labels for $T$, and finally accepts or rejects based on the accuracy of these labels. As a corollary of Theorem~\ref{thm:mainthm1}, an efficient quantum learner can pass this challenge, while no efficient classical learner can pass it assuming the classical hardness of $\DLP$.

\bigskip
{\noindent\bf Efficient learnability with $\QKE$}\\[1mm]
We now turn our attention to general-purpose quantum learning algorithms which only requires classical access to data and in principle can be applied to a wide range of learning problems. Examples include quantum neural networks~\cite{mitarai2018quantum,farhi2018classification,grant2018hierarchical,schuld2020circuit,Benedetti_2019}, generative models~\cite{Gao2018quantum}, and kernel methods~\cite{Havlicek2019,Schuld2019quantum}. The main challenge of proving a quantum advantage for these algorithms is that they may not be able to utilize the full power of quantum computers, and it is unclear if the set of problems that they can solve is beyond the capability of classical algorithms. Indeed, previous analysis of these quantum algorithms only establish evidence that parts of the algorithm cannot be efficiently simulated classically~\cite{Gao2018quantum,Havlicek2019,Schuld2019quantum}, which does not guarantee that the algorithms can solve classically hard learning problems.

Recall that we have constructed a learning problem that is as hard as the discrete log problem. This implies classical intractability assuming the hardness of discrete log, while also assuring that the problem is within the power of quantum computers due to Shor's algorithm. This provides the basis to solving our main challenge -- we now show that the concept class $\mc C$ can be efficiently learned by our support vector machine algorithm with quantum kernel estimation ($\SVM$-$\QKE$). This formally establishes our intuition that quantum feature maps can recognize patterns that are unrecognizable by classical algorithms, even when the quantum classifier is inherently noisy due to finite sampling statistics. We have therefore established an end-to-end quantum advantage for quantum kernel methods.

The core component in our algorithm that leads to its ability to outperform classical learners is the \emph{quantum feature map}. For learning the concept class $\mc C$, the feature map is constructed prior to seeing the training samples and has the following form,
\begin{equation}
    x\mapsto\ket{\phi(x)}=\frac{1}{\sqrt{2^k}}\sum_{i=0}^{2^k-1}\ket{x\cdot g^i},
\end{equation}
which maps a classical data point $x\in\ints_p^*$ to a $n$-qubit quantum state $\ketbra{\phi(x)}$, and $k=n-t\log n$ for some constant $t$. This family of states was first introduced in~\cite{Aharonov2007adiabatic} to study the complexity of quantum state generation and statistical zero knowledge, where it is shown that  $\ket{\phi(x)}=U(x)\ket{0^n}$ can be efficiently prepared on a fault tolerant quantum computer by a circuit $U(x)$ which uses Shor's algorithm as a subroutine, c.f.\ Appendix~\ref{app:qfeaturemap}.

In learning algorithms, feature maps play the role of pattern recognition: the intrinsic labeling patterns for data, which are hard to recognize in the original space (Fig.~\ref{fig:data}b), become easy to identify once mapped to the feature space. Our feature map indeed achieves this by mapping low dimensional data vectors to the Hilbert space with exponentially large dimension. For each concept $f_s\in\mc C$, we show that there exists a separating hyperplane $w_s = \ketbra{\phi_s}$ in feature space. That is, for~$+1$ labeled examples, we have $\Tr[w_s\ketbra{\phi(x)}] = |\braket{\phi_s}{\phi(x)}|^2=1/\poly(n)$ , while for $-1$ labeled examples we have $|\braket{\phi_s}{\phi(x)}|^2=0$ with probability $1 - 1/\poly(n)$ (Fig.~\ref{fig:data}c). If we think of~$w_s$ as the normal vector of a hyperplane in feature space associated with the Hilbert-Schmidt inner product, then this property suggests that +1/-1 labeled examples are separated by this hyperplane by a large margin. This \emph{large margin property} is the fundamental reason that our algorithm can succeed: it suggests that to correctly classify the data samples, it suffices to find a good linear classifier in feature space, while also guaranteeing that such a good linear classifier~exists. 

The idea of applying a high-dimensional feature map to reduce a complex pattern recognition problem to linear classification is not new, and has been the foundation of a family of supervised learning algorithms called support vector machines ($\SVM$s)~\cite{boser1992training,vapnik2013nature}. Consider a general feature map $\phi:\calX\to\mc H$ that maps data to a feature space~$\mc H$ associated with an inner product $\langle\cdot,\cdot\rangle$. To find a linear classifier in $\mc H$, we consider the convex quadratic program  
\begin{align}
\label{eq:mainprimal}
\min_{w,\xi}\frac{1}{2}\|w\|_2^2 + \frac{\lambda}{2}\sum_{i=1}^m \xi_i^2 \hspace{2mm}  \text{ s.t. } y_i\cdot \langle \phi(x_i), w \rangle\geq 1-\xi_i
\end{align}
where $\xi_i\geq 0$. Here $\lambda>0$ is a constant, $w$ is a hyperplane in $\mc H$ which defines a linear classifier $y=\sign\left( \langle \phi(x),w \rangle\right)$, and $\xi_i$ are slack variables used in the soft margin constraints. Intuitively, this program optimizes for the hyperplane that maximally separates +1/-1 labeled data. Note that \eqref{eq:mainprimal} is efficient in the dimension of $\mc H$. However, once we map to a high-dimensional feature space, it takes exponential time to find the optimal hyperplane. The main insight which leads to the success of $\SVM$s is that this problem can be solved by running the dual program of Eq.~\eqref{eq:mainprimal}
\begin{equation}\label{eq:maindual}
    \max_{\alpha\geq 0}\sum_{i=1}^m\alpha_i-\frac{1}{2}\sum_{i,j=1}^m\alpha_i\alpha_j y_i y_j K(x_i,x_j)-\frac{1}{2\lambda}\sum_{i=1}^m\alpha_i^2,
\end{equation}
where $K(x_i,x_j)= \langle \phi(x_i),\phi(x_j)\rangle$ is the \emph{kernel matrix}. This dual program, which returns a linear classifier defined as $y=\sign\left(\sum_{i=1}^m \alpha_i y_i K(x,x_i)\right)$, is equivalent to the original program as guaranteed by strong duality. Effectively, this means that we can do optimization in the high-dimensional feature space efficiently, as long as the kernel $K(x_i,x_j)$ can be efficiently computed.

\begin{figure}[t]
    \centering
    \includegraphics[width=0.5\linewidth]{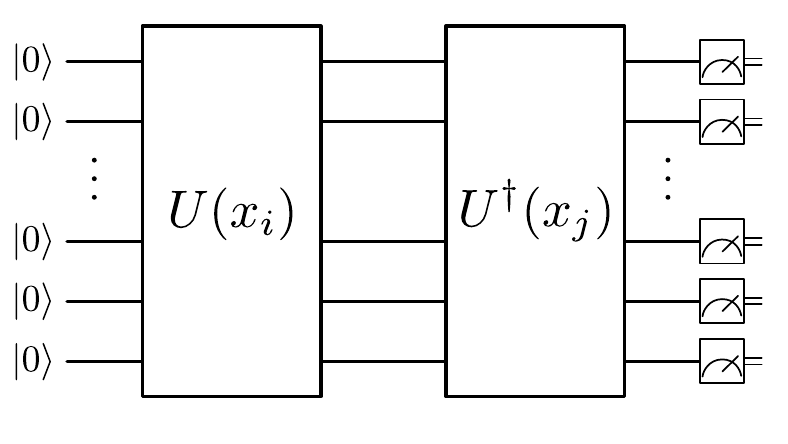}
    \caption{Quantum kernel estimation. A quantum feature map $x\mapsto\ket{\phi(x)}$ is represented by a circuit, $\ket{\phi(x)}=U(x)\ket{0^n}$. Each kernel entry $K(x_i,x_j)$ is obtained using a quantum computer by running the circuit $U^\dag (x_j)U(x_i)$ on input $\ket{0^n}$, and then estimate $\left|\expval{U^\dag (x_j)U(x_i)}{0^n}\right|^2$ by counting the frequency of the $0^n$ output.}
    \label{fig:kernel}
\end{figure}

The same insight can be applied to our quantum feature map: to utilize the full power of the quantum feature space, it suffices to compute the inner products $|\braket{\phi(x_i)}{\phi(x_j)}|^2$ between the feature states. To estimate such an inner product using a quantum computer, we simply apply $U^\dag (x_j)U(x_i)$ on input $\ket{0^n}$, and measure the probability of the $0^n$ output (see Fig.~\ref{fig:kernel}). We call such a procedure quantum kernel estimation ($\QKE$). The overall procedure for learning with quantum feature map is now clear. On input a set of $m$ labeled training examples $S$, run $\QKE$ to obtain the $m\times m$ kernel matrix, then run the dual $\SVM$ Eq.~\eqref{eq:maindual} on a classical computer to obtain the solution~$\alpha$. To classify a new example $x$, run $\QKE$ to obtain $K(x,x_i)$ for each $i=1,\ldots, m$, then~return 
\begin{equation}
    y=\sign\left(\sum_{i=1}^m \alpha_i y_i K(x,x_i)\right).
\end{equation}
Throughout the entire $\SVM$-$\QKE$ algorithm, $\QKE$ is the only subroutine that requires a quantum computer, while all other optimization steps can be performed classically. See Appendix~\ref{section:tutorial} for a detailed description of the algorithm.

Despite the seemingly direct analogy between quantum and classical feature maps, one important aspect of $\QKE$ makes the analysis of our quantum algorithm fundamentally different from classical $\SVM$s. Note that estimating the output probability of a quantum computer is inherently noisy due to finite sampling statistics, even when the quantum computer is fully error corrected. In $\QKE$, this finite sampling error can be modeled as i.i.d.\ additive errors for each kernel entry, with mean 0 and variance $\frac{1}{R}$, where $R$ is the number of measurement shots for each kernel estimation circuit. Our main result rigorously establishes the performance guarantee of $\SVM$-$\QKE$, which remains robust under this noise model.

\begin{theorem}\label{thm:mainthm2}
The concept class $\mc C$ is efficiently learnable by $\SVM$-$\QKE$. More specifically, for any concept $f\in~\mc C$, the $\SVM$-$\QKE$ algorithm returns a classifier with test accuracy at least 0.99 in polynomial time, with probability at least $2/3$ over the choice of random training samples and over noise in $\QKE$ estimation.
\end{theorem}

The main idea in our proof, c.f.\ Appendix~\ref{app:qkelearnability} and~\ref{section:generalizationbound}, is to connect the large margin property to existing results on the generalization of soft margin classifiers~\cite{anthony_bartlett_2000,bartlett1998prediction,Shawe-Taylor1998structural,BST1999,Shawe-Taylor2002generalization}. There, it is shown that if the learning algorithm finds a hyperplane $w$ that has a large margin on the training set, then the linear classifier $y=\sign\left(\langle \phi(x),w\rangle\right)$ has high accuracy with high probability. To see how we can apply these results, recall that in the large margin property, we have established that there exists a hyperplane $w^*$ with a large margin on the training set. Therefore, as the $\SVM$ program optimizes for the hyperplane with largest margin, it is guaranteed to find a good hyperplane $w$, although not necessarily the same as $w^*$. Applying the generalization bounds to this $w$ gives us the desired result.

As discussed above, the missing piece in the proof sketch is to show that the performance of $\SVM$-$\QKE$ remains robust with additive noise in the kernel. In the following we prove noise robustness by introducing two additional results. First, we show that the dual $\SVM$ program (Eq.~\eqref{eq:maindual}) is robust, i.e., when the kernel used in~\eqref{eq:maindual} has a small additive perturbation, then the solution returned by the program also has a small perturbation. This follows from strong convexity of \eqref{eq:maindual} and standard perturbation analysis of positive definite quadratic programs~\cite{Daniel1973}. This result implies that the hyperplane $w'$ obtained by the noisy kernel is close to the noiseless solution $w$ with high probability. Second, we show that when $w'$ is close to $w$, the linear classifier obtained by $w'$ has high accuracy. This seemingly simple statement is not trivial, as the sign function is sensitive to noise. That is, if $\langle \phi(x),w\rangle$ is very close to 0, then a small perturbation in $w$ could change its sign. We provide a solution to this problem by proving a stronger generalization bound. We show that if a hyperplane $w$ has a large margin on the training set, then not only does $\langle \phi(x),w\rangle$ have the correct sign, it is also bounded away from 0 with high probability. Therefore, when the noisy solution $w'$ is close to $w$, $\langle \phi(x),w'\rangle$ also has the correct sign with high probability. Combining these two results with the proof sketch, we have the full proof of Theorem~\ref{thm:mainthm2}.

\bigskip
{\noindent\bf Conclusions and outlook}\\[1mm]
We show that learning with quantum feature maps provides a way to harness the computational power of quantum mechanics in machine learning problems. This idea leads to a simple quantum machine learning algorithm that makes no additional assumptions on data access and has rigorous and robust performance guarantees. While the learning problem we have presented here that demonstrates an exponential quantum speed-up is not practically motivated, our result sets a positive theoretical foundation for the search of practical quantum advantage in machine learning. An important future direction is to construct quantum feature maps that can be applied to practical machine learning problems that are classically challenging. The results we have established here can be useful for the theoretical analysis of such proposals.

An important advantage of the $\SVM$-$\QKE$ algorithm, which only uses quantum computers to estimate kernel entries, is that error-mitigation techniques can be applied~\cite{temme2017error,li2017efficient,kandala2019error} when the feature map circuit is sufficiently shallow. Our robustness analysis gives hope that an error-mitigated quantum feature map can still maintain its computational power. Finding quantum feature maps that are sufficiently powerful and shallow is therefore the stepping stone towards obtaining a quantum advantage in machine learning on near-term devices.

\begin{acknowledgments}
We thank Sergey Bravyi and Robin Kothari for helpful comments and discussions. Y.L. was supported by Vannevar Bush faculty fellowship N00014-17-1-3025 and DOE QSA grant \#FP00010905. Part of this work was done when Y.L.~was a research intern at IBM. S.A.~and K.T.~acknowledge support from the MIT-IBM  Watson  AI  Lab  under  the  project {\it Machine  Learning  in  Hilbert  Space}, the  IBM  Research  Frontiers  Institute and the ARO Grant W911NF-20-1-0014.
\end{acknowledgments}

\bibliographystyle{apsrev4-2}
\bibliography{ref}

\onecolumngrid
\appendix

\section{Supervised learning and the discrete log problem}\label{app:definitions}
We work in the same setting as in standard computational learning theory~\cite{kearns1990computational,kearns1994introduction}. The \emph{data space} $\calX\subseteq \mathbb{R}^d$ is a fixed subset of $d$-dimensional Euclidean space. In this paper we will be concerned with \emph{distribution-dependent} learning, i.e., we fix our \emph{data distribution} to be the uniform distribution over $\calX$. A \emph{concept class} $\mc C$ is a set of functions that maps data vectors to binary labels, i.e., every $f\in \mc C$ is a function $f:\calX\to\{-1,1\}$.

A learning algorithm is given a set of \emph{training samples} $S=\{(x_i,f(x_i))\}_{i=1}^m$, where each $x_i$ is independently drawn from the uniform distribution over $\calX$, and $f\in\mc C$ is an unknown concept in the concept class.
The goal of the learning algorithm is to return a classifier $f^*$ that runs in polynomial time. The \emph{test accuracy} of the learned classifier is defined as the probability of agreeing with the unknown concept,
\begin{equation}
    \mathrm{acc}_{f}(f^*)=\Pr_{x\sim \calX}\left[f^*(x)=f(x)\right].
\end{equation}

\begin{definition}[Efficient learning of $\Cc$]\label{def:efficientlearning}
Let $\calX\subseteq \mathbb{R}^d$. A concept class $\mc C\subseteq \{f:\calX\rightarrow \pmset{}\}$ is \emph{efficiently learnable}, if there exists a learning algorithm $\mc A$ that satisfies the following: for every $f\in \Cc$, algorithm $\mc A$ takes as input $\poly(d)$ many training samples $S$ and with probability at least~$2/3$ (over the choice of random training samples and randomness of the algorithm), outputs a classifier in time $\poly(d)$ that achieves 99\% test accuracy.
\end{definition}

The concept class we construct for showing our quantum advantage is based on the \emph{discrete log problem} ($\DLP$) which we define first:
\begin{quote}
    $\DLP$: given a prime $p$, a primitive element $g$ of $\ints_p^*=\{1,2,\dots,p-1\}$, and $y\in \ints_p^*$, find $x\in \ints_p^*$ such that $g^x\equiv y\pmod p$.
\end{quote}
For a fixed $p,g$, we let $\DLP(p,g)$ be the discrete log problem with inputs $y\in \ints_p^*$. The input to the  $\DLP$ problem can be described by $n=\lceil\log_2 p\rceil$ bits. It is shown~\cite{blum1984how} that $\DLP$ is reducible to the following \emph{decision problem} $\DLP_{\frac{1}{2}}$:
\begin{itemize}
    \item \textbf{Input:} prime $p$, generator $g$ of $\ints_p^*$, $y\in \ints_p^*$.
    \item \textbf{Output:} 1 if $\log_g y\leq\frac{p-1}{2}$ and 0 otherwise.
\end{itemize}

\begin{lemma}[{\cite[Theorem~3]{blum1984how}}]\label{lemma:dlpsmallbiashardness}
For every prime $p$ and generator $g$, if there exists a polynomial time algorithm that correctly decides $\DLP_{\frac{1}{2}}(p,g)$ for at least $\frac{1}{2}+\frac{1}{\poly(n)}$ fraction of the inputs $y\in \ints_p^*$, then there exists a polynomial time algorithm for $\DLP(p,g)$.
\end{lemma}

Furthermore,~\cite{blum1984how} showed that $\DLP$ can be further reduced to the following promise discrete logarithm problem~$\DLP_c$ for $\frac{1}{\poly(n)}\leq c\leq\frac{1}{2}$:
\begin{itemize}
    \item \textbf{Input:} prime $p$, generator $g$ of $\ints_p^*$, $y\in \ints_p^*$.
    \item \textbf{Promise:} $\log_g y\in[1,c(p-1)]$ or $[\frac{p-1}{2}+1,\frac{p-1}{2}+c(p-1)]$
    \item \textbf{Output:} $-1$ if $\log_g y\in [1,c(p-1)]$ and $+1$ if $\log_g y\in [\frac{p-1}{2}+1,\frac{p-1}{2}+c(p-1)]$.
\end{itemize}

\begin{lemma}[\cite{blum1984how}]\label{lemma:dlpreducetopromise}
For every prime $p$, generator $g$ and $\frac{1}{\poly(n)}\leq c \leq\frac{1}{2}$, if there exists a polynomial time algorithm for $\DLP_{c}(p,g)$, then there exists a polynomial time algorithm for $\DLP(p,g)$.
\end{lemma}
The proof of this fact is implicitly implied by the proof of Lemma 3 and Theorem 3 in~\cite{blum1984how}.

\section{A concept class reducible to discrete log}\label{app:classicalhardness}
In this section we construct a concept class, wherein learning the concept class is as hard as solving $\DLP_{\frac{1}{2}}(p,g)$. Therefore, assuming the hardness of $\DLP(p,g)$, no classical polynomial time algorithm can learn this concept class. On the other hand, the learning problem is a simple clustering problem in 1D after taking the  discrete logarithm, and therefore is easy to learn using a quantum~computer.

We work in the setting introduced in the previous section, where we use standard definitions (see Definition~\ref{def:efficientlearning}) from computational learning theory, and we assume a fixed $p,g$ such that computing discrete log in $\ints_p^*$ is classically hard. Our concept class is defined as~follows.

\begin{definition}[Concept class]\label{def:conceptclass}
We define a concept class over the data space $\calX=\ints_p^*\subseteq\{0,1\}^n$, where $n=\lceil\log_2 p\rceil$. For any $s\in \ints_p^*$, define a concept $f_s:\ints_p^*\to \pmset{}$ as
    \begin{equation}
        f_s(x)=\begin{cases}+1,&\text{if }\log_g x\in[s,s+\frac{p-3}{2}],\\ -1, &\text{else.}\end{cases}
    \end{equation}
    Note that in interval $[s,s+\frac{p-3}{2}]$, $s+i$ denotes addition within $\ints_p^*$. By definition, $f_s$ maps half the elements in $\ints_p^*$ to $+1$ and half of them to $-1$. The concept class is defined as $\mc C=\{f_s\}_{s\in \ints_p^*}$.
\end{definition}

A target concept in $\mc C$ can be specified by choosing an element $s\in \ints_p^*$, which can be understood as the ``secret key" for the concept $f_s$.
We can also efficiently generate training samples for every~concept.
\begin{lemma}\label{lemma:efficientdatageneration}
    For every concept $f\in\mc C$, there exists an efficient classical algorithm that can generate samples $(x,f(x))$, where $x$ is uniformly random in $\ints_p^*$.
\end{lemma}
\begin{proof}
    To generate a sample $(x,f_s(x))$ for a concept $f_s$, first generate a random $y\sim \ints_p^*$, and let 
    \begin{equation}
        b=\begin{cases}+1,&\text{if }y\in[s,s+\frac{p-3}{2}],\\ -1, &\text{else.}\end{cases}
    \end{equation}
    Then return $(g^y,b)$. Since $g^y$ is also uniformly distributed in $\ints_p^*$, this procedure correctly generates a sample from the data distribution.
\end{proof}

Using the quantum algorithm for discrete logarithm problem~\cite{Shor1997polynomial,Mosca2004exact}, the concept class~$\mc C$ is polynomially learnable in $\BQP$ (in fact with probability $1$). On the other hand, the result of Blum and Micali~\cite{blum1984how} implies that no efficient classical algorithm can do better than random guessing.

\begin{theorem}\label{thm:conceptclasshardness}
The concept class $\mc C$ is efficiently learnable in $\BQP$. On the other hand, suppose there exists an efficient classical algorithm that, for every concept $f\in\mc C$, can achieve $\frac{1}{2}+\frac{1}{\poly(n)}$ test accuracy, with probability at least $2/3$ over the choice of random training samples and randomness of the algorithm. Then there exists an efficient classical algorithm for $\DLP$.
\end{theorem}

\begin{remark}\label{rmk:classicalhardness}
In the following we prove a stronger statement for classical hardness. We show that assuming the classical hardness of $\DLP$, no efficient classical algorithm can achieve $\frac{1}{2}+\frac{1}{\poly(n)}$ test accuracy with probability $\frac{2}{3}$ for \emph{any} concept in the concept class $\mc C$.
\end{remark}
\begin{proof}
We first show quantum learnability. For every concept $f_s\in\mc C$, use a quantum computer to take the discrete logarithm of the classical training data samples. Then, after taking the discrete logarithm, the training data samples are clustered in two intervals: $[s,s+\frac{p-3}{2}]$ with label $+1$, and $[s+\frac{p-1}{2},s+p-2]$ with label $-1$, for the unknown $s\in\ints_p^*$ (which defines the unknown concept $f_s$). For a new data sample $x$, use a quantum computer to take its discrete log. Then, compute the average distance $d_+/d_-$ between $\log_g x$ and the $+1$/$-1$ labeled clusters, respectively. Assign label to $x$ based on which cluster is closer to $\log_g x$. This algorithm can achieve 99\% accuracy for any concept $f_s$, with high probability over random training samples. We omit the detailed proof here.

To show classical hardness as stated in Remark~\ref{rmk:classicalhardness}, consider an arbitrary concept $f_s$ and polynomially many training samples. By Lemma~\ref{lemma:efficientdatageneration} this can be generated classically in polynomial time. By assumption, an efficient classical algorithm $\mc A$ can learn this concept with $\frac{1}{2}+\frac{1}{\poly(n)}$ accuracy (call $\mc A$ a good classifier if it satisfies this), with probability at least $2/3$.  We use this learned classifier to solve $\DLP_{\frac{1}{2}}(p,g)$.

\noindent\textbf{Algorithm $\mc A'$ for $\DLP_{\frac{1}{2}}(p,g)$:}
\begin{enumerate}
    \item On input $y\in \ints_p^*$ such that $\log_g y\in[1,\frac{p-1}{2}]$ or $[\frac{p-1}{2}+1,p-1]$.
    \item Send $y\cdot g^{s-1}$ to the classifier $\mc A$, decide $\log_g y\in[1,\frac{p-1}{2}]$ if the classifier returns $+1$, and decide $\log_g y\in[\frac{p-1}{2}+1,p-1]$ if the classifier returns $-1$.
\end{enumerate}
To see that this procedure correctly decides $\DLP_{\frac{1}{2}}(p,g)$ with a non-trivial bias, for a good classifier~$\mc A$ we have
\begin{nalign}
    &\Pr_{y\sim\ints_p^*}\left[\mc A'\text{ correctly decides }\DLP_{\frac{1}{2}}(p,g)\text{ on input }y\right]\\
    &=\Pr_{y\sim\ints_p^*}\left[\mc A\text{ correctly classifies }y\cdot g^{s-1}\text{ for concept }f_s\right]\\
    &=\Pr_{y\sim\ints_p^*}\left[\mc A\text{ correctly classifies }y\text{ for concept }f_s\right]\\
    &=\frac{1}{2}+\frac{1}{\poly(n)}.
\end{nalign}
By Lemma~\ref{lemma:dlpsmallbiashardness}, once we have an algorithm that can correctly decide $\DLP_{\frac{1}{2}}(p,g)$ on $\frac{1}{2}+\frac{1}{\poly(n)}$ fraction of the inputs, it can be used to solve $\DLP(p,g)$ with high success probability. Finally by a simple union bound, we have a polynomial time algorithm that with high probability solves~$\DLP(p,g)$. 
\end{proof}

An advantage of our supervised learning task is that it is efficiently verifiable by a classical verifier. Consider the following challenge:
\begin{enumerate}
    \item A classical verifier picks a random concept $f_s\sim\mc C$ (it can do so by choosing a uniformly random $s\sim \ints_p^*$). Then, generate polynomial-sized samples $(S,T)$ where the data labels in~$T$ are removed.
    \item The verifier sends $(S,T)$ to a prover, and the prover returns a set of $\{-1,1\}$ labels for $T$.
    \item The verifier accepts if more than 99\% of the labels returned by the prover are correct.
\end{enumerate}

Say a prover passes the challenge if the verifier accepts with probability at least $2/3$. Our hardness result implies the following Corollary:
\begin{corollary}
    There exists a $\BQP$ prover that can pass the above challenge. Assuming the classical hardness of $\DLP$, no polynomial-time classical prover can pass the above challenge.
\end{corollary}

\section{Support vector machine and the quantum kernel estimation algorithm}\label{section:tutorial}
\subsection{Support vector machines}
We give a brief overview of support vector machine and the quantum kernel estimation algorithm~\cite{Havlicek2019,Schuld2019quantum}. Along the way, we also establish properties that are useful for the analysis of our algorithm in the next section. We refer to Ref.~\cite{Burges1998,Smola2004} for a more detailed introduction to support vector machines.

A support vector machine ($\SVM$) is a classification algorithm that takes as input a set of training samples $S=\{(x_1,y_1),\dots,(x_m,y_m)\}$ where $x_i\in\mathbb{R}^d$, $y\in \pmset{}$ and in time $\poly(d,m)$ (assume that the training set has polynomial size, $m=\poly(d)$) returns a set of parameters $(w,b)\in \mathbb{R}^d\times \mathbb{R}$ which define a linear classifier $f^*:\calX\to \pmset{}$ as follows
\begin{equation}
    y_{\text{pred}}=f^*(x)=\sign\left(\langle w,x\rangle+b\right),
\end{equation}
where $\langle w,x\rangle=\sum_ix_i w_i$. For a data vector $x$ with true label $y$, it is easy to see that the classifier is correct on this point \emph{if and only if} $y\left(\langle w,x\rangle+b\right)>0$. We say a training set $S$ is \emph{linearly separable} if there exists $(w,b)$ such~that
\begin{equation}
    y_i\left(\langle w,x_i\rangle+b\right)>0,\,\,\,\,\text{ for every } (x_i,y_i)\in S,
\end{equation}
and such a $(w,b)$ is called a \emph{separating hyperplane} for $S$ in $\mathbb{R}^d$. When the training set is linearly separable, the $\SVM$ algorithm can efficiently find a separating hyperplane by running the so-called \emph{hard margin primal~program}
\begin{nalign}\label{eq:svmhardmarginprimal}
     \min_{w,b}\,\,\,\,\,\,\,\,&\frac{1}{2}\|w\|_2^2\\
     \text{s.t.}\,\,\,\,\,\,\,\,&y_i\left(\langle x_i,w\rangle+b\right)\geq 1,
\end{nalign}
a convex quadratic program whose optimal solution can be obtained in polynomial time. One important property of $\SVM$ is that it is a maximum margin classifier.
For a general unnormalized hyperplane $(w,b)$, define its \emph{normalized margin} on a training data $(x,y)$ as
\begin{equation}
    \hat{\gamma}_{(w,b)}(x,y)=\frac{1}{\|w\|_2}y\left(\langle w,x\rangle+b\right)
\end{equation}
and let $\gamma_{(w,b)}(x,y)=y\left(\langle w,x\rangle+b\right)$ denote the unnormalized margin. It is easy to see that Eq.~\eqref{eq:svmhardmarginprimal} returns a classifier that maximizes
\begin{equation}
    \min_{(x,y)\in S}\hat{\gamma}_{(w,b)}(x,y),
\end{equation}
which is the minimum distance from any training point to the hyperplane. The general intuition that $\SVM$ maximizes the margin is useful for understanding the generalization bounds that we will prove in the next section.

However, for most ``practical" purposes, assuming $S$ is linearly separable is a strong assumption. Additionally, when the training set $S$ is not linearly separable, Eq.~\eqref{eq:svmhardmarginprimal} does not have a feasible solution. To find a good linear classifier with the presence of outliers, we introduce the \emph{soft margin primal~program}
\begin{nalign}\label{eq:generalsoftmarginprimal}
     \min_{w,\xi,b}\,\,\,\,\,\,\,\,&\frac{1}{2}\|w\|_2^2 + \frac{\lambda}{2}\sum_i\xi_i^p\\
     \text{s.t.}\,\,\,\,\,\,\,\,&y_i\left(\langle x_i,w\rangle+b\right)\geq 1-\xi_i\\
     &\xi_i\geq 0,
\end{nalign}
where $\xi_i$ are the \emph{slack variables} introduced to relax the margin constraints, with an additional penalty term $\frac{\lambda}{2}\sum_i\xi_i^p$. For any positive integer $p$, Eq.~\eqref{eq:generalsoftmarginprimal} is a convex program and is feasible. In this work, we focus on choosing $p=2$, which becomes a quadratic program. In practice it is also common to use $p=1$.

For the $p=2$ case, we further derive the Wolfe dual program of \eqref{eq:generalsoftmarginprimal} based on Lagrangian duality, resulting in the \emph{$L2$ soft margin dual program}
\begin{nalign}\label{eq:generalL2dual}
     \max_{\alpha}\,\,\,\,\,\,\,\,&\sum_i\alpha_i-\frac{1}{2}\sum_{i,j}\alpha_i\alpha_j y_i y_j \langle x_i,x_j\rangle-\frac{1}{2\lambda}\sum_i\alpha_i^2\\
     \text{s.t.}\,\,\,\,\,\,\,\,&\alpha_i\geq0\\
     &\sum_i\alpha_i y_i=0.
\end{nalign}
Here the primal Lagrangian is given by
\begin{equation}
    L_P=\frac{1}{2}\|w\|_2^2+\frac{\lambda}{2}\sum_i\xi_i^2-\sum_i\alpha_i\left(y_i\left(\langle x_i,w\rangle+b\right)- 1+\xi_i\right)-\sum_i \mu_i\xi_i,
\end{equation}
and the primal and dual optimal solutions can be connected via the Karush-Kuhn-Tucker (KKT) conditions
\begin{nalign}\label{eq:kktconditions}
    &w=\sum_i \alpha_i y_i x_i\\
    &\sum_i\alpha_i y_i=0\\
    &\lambda\xi_i-\alpha_i-\mu_i=0\\
    &\alpha_i\left(y_i\left(\langle x_i,w\rangle+b\right)- 1+\xi_i\right)=0\\
    &\mu_i\xi_i=0\\
    &\alpha_i\geq 0\\
    &\mu_i\geq 0\\
    &\xi_i\geq 0\\
    &y_i\left(\langle x_i,w\rangle+b\right)-1+\xi_i\geq 0.
\end{nalign}
An immediate corollary of Eq.~\eqref{eq:kktconditions} is
\begin{equation}
    \alpha_i=\lambda\xi_i
\end{equation}
at the optimal solution. This means that when $\lambda$ is a constant, the Lagrangian multipliers $\alpha_i$ is proportional to the slack variables $\xi_i$. This is a useful property for our analysis later. 
In addition, the bias parameter $b$ can be determined by the equality $y_i\left(\langle x_i,w\rangle+b\right)- 1+\xi_i=0$ for any $\alpha_i\neq 0$.

Finally, it will be convenient for us to work with optimizations without the bias parameter $b\in \mathbb{R}$.  
We show that we can assume this is without loss of generality, as we can add one extra dimension $\tilde{x}= (x,1)/\sqrt{2}$ to the data vectors so that the bias parameter is absorbed into $w$. In this case, the $L2$ soft margin dual program becomes
\begin{nalign}\label{eq:L2dualnobias}
     \max_{\alpha}\,\,\,\,\,\,\,\,&\sum_i\alpha_i-\frac{1}{4}\sum_{i,j}\alpha_i\alpha_j y_i y_j \left(\langle x_i,x_j\rangle+1\right)-\frac{1}{2\lambda}\sum_i\alpha_i^2\\
     \text{s.t.}\,\,\,\,\,\,\,\,&\alpha_i\geq0.
\end{nalign}
Here we used the fact that $\langle\tilde{x}_i,\tilde{x}_j\rangle=\frac{1}{2}\left(\langle x_i,x_j\rangle+1\right)$, and notice that the equality constraint $\sum_i\alpha_i y_i=0$ is removed. This is because of the new KKT conditions
\begin{nalign}\label{eq:kktconditionsnobias}
    &w=\sum_i \alpha_i y_i x_i\\
    &\lambda\xi_i-\alpha_i-\mu_i=0\\
    &\alpha_i\left(y_i\left(\langle x_i,w\rangle\right)- 1+\xi_i\right)=0\\
    &\mu_i\xi_i=0\\
    &\alpha_i\geq 0\\
    &\mu_i\geq 0\\
    &\xi_i\geq 0\\
    &y_i\left(\langle x_i,w\rangle\right)-1+\xi_i\geq 0,
\end{nalign}
where $\alpha_i=\lambda\xi_i$ still hold at optimality.

\subsection{Non-linear classification}
In this section we generalize support vector machines for \emph{non-linear} classification, i.e., we map the $d$-dimensional data vectors into a $n$-dimensional feature space ($n\gg d$) via a \emph{feature map}:
\begin{equation}
    \phi:\calX\to \mathbb{R}^n,
\end{equation}
where we assume that $\phi$ is normalized, i.e., it maps a unit vector to a unit vector. The feature map is chosen prior to seeing the training data. Notice that in the dual program, training data is only accessed via the \emph{kernel}  matrix $ K\in \mathbb{R}^{m\times m}$ (recall that $m$ denotes the number of training samples) defined as
\begin{equation}
K(x_i,x_j)=\langle \phi(x_i),\phi(x_j)\rangle.
\end{equation}
Therefore, it's possible to work with an exponentially large (or even infinite-dimensional) feature space (i.e., when $n$ is large), as long as the kernel is computable in $\poly(d)$ time. 

In addition, for any feature map $\phi_0:\calX\to \mathbb{R}^n$, we can always use a new feature map $\phi:\calX\to \mathbb{R}^{n+1}$ such that $\phi(x)=(\phi_0(x),1)/\sqrt{2}$, which allows us to remove the bias parameter $b$. This can be done via changing the kernel as $K(x_i,x_j)=\frac{1}{2}\left(K_0(x_i,x_j)+1\right)$ as shown in the previous section. Therefore, with a suitable kernel transformation, we can run the following dual program without loss of generality:
\begin{nalign}\label{eq:L2dualnobias2}
     \max_{\alpha}\,\,\,\,\,\,\,\,&\sum_i\alpha_i-\frac{1}{2}\sum_{i,j}\alpha_i\alpha_j y_i y_j K(x_i,x_j)-\frac{1}{2\lambda}\sum_i\alpha_i^2\\
     \text{s.t.}\,\,\,\,\,\,\,\,&\alpha_i\geq0.
\end{nalign}
Let $Q\in\mathbb{R}^{m\times m}$ be a matrix such that $Q_{ij}=y_i y_j K(x_i,x_j)$. Then we can write the dual program in vectorized form
\begin{nalign}\label{eq:L2dualnobiasvectorized}
     \max_{\alpha}\,\,\,\,\,\,\,\,&1^T\alpha-\frac{1}{2}\alpha^T\left(Q+\frac{1}{\lambda}\mathbb{I}\right)\alpha\\
     \text{s.t.}\,\,\,\,\,\,\,\,&\alpha\geq0.
\end{nalign}
It is easy to see that for every $\lambda>0$,  Eq.~\eqref{eq:L2dualnobiasvectorized} is a strongly convex quadratic program and has a unique optimal solution. After training is finished (i.e., we obtain training samples, compute the kernel $K$ and solve Eq.~\eqref{eq:L2dualnobiasvectorized}), when a learner is presented with a  new test example $x$, the classifier~returns
\begin{equation}
    y_{\text{pred}}=\sign\left(\langle w,\phi(x)\rangle\right)=\sign\left(\sum_i \alpha_i y_i K(x_i,x)\right),
\end{equation}
where we used the KKT condition $w=\sum_i\alpha_i y_i \phi(x_i)$ (which can be derived by simply replacing $x_i$ with $\phi(x_i)$ in Eq.~\eqref{eq:kktconditionsnobias}). So a classifier needs to evaluate the kernel function on the new test example and output $y_{\text{pred}}$.

\subsection{Quantum kernel estimation}
Different from the above standard approaches, the kernel used in our quantum algorithm is constructed by a quantum feature map. 
The main idea in the \emph{quantum kernel estimation} algorithm~\cite{Havlicek2019,Schuld2019quantum} is to map classical data vectors into quantum states:
\begin{equation}
    x\to \ketbra{\phi(x)},
\end{equation}
where we use the density matrix representation to avoid global phase. Then, the kernel function is the Hilbert-Schmidt inner product between density matrices,
\begin{equation}
    K(x_i,x_j)=\Tr\big[\ketbra{\phi(x_i)}\cdot \ketbra{\phi(x_j)}\big]=\left|\braket{\phi(x_i)}{\phi(x_j)}\right|^2.
\end{equation}
This \emph{quantum feature map} is implemented via a quantum circuit parameterized by $x$,
\begin{equation}
    \ket{\phi(x)}=U(x)\ket{0^n},
\end{equation}
where we assume the feature map uses $n$ qubits. Therefore, to obtain the kernel function
\begin{equation}
    K(x_i,x_j)=\left|\bra{0^n}U^\dag(x_i)U(x_j)\ket{0^n}\right|^2,
\end{equation}
we can run the quantum circuit $U^\dag(x_i)U(x_j)$ on input $\ket{0^n}$, and measure the probability of the~$0^n$~output.

\begin{algorithm}[H]
\caption{Support vector machine with quantum kernel estimation ($\SVM$-$\QKE$ training)}\label{alg:qketraining}
\textbf{Input:} a training set $S=\{(x_i,y_i)\}_{i=1}^m$\\
\textbf{Output:} $\alpha_1,\ldots,\alpha_m$ (solution to the $L2$ soft margin dual program~\eqref{eq:L2dualnobiasvectorized})
\begin{algorithmic}[1]
\For {$i=1\dots m$}
    \State $K_0(x_i,x_i):=1$
\EndFor
\For {$i=1\dots m$}\Comment{quantum kernel estimation}
    \For {$j=i+1\dots m$}
        \State Apply $U^\dag(x_i)U(x_j)$ on input $\ket{0^n}$
        \State Measure the output probability of $0^n$ with $R$ shots, denoted as $p$
        \State $K_0(x_i,x_j)=K_0(x_j,x_i):=p$
\EndFor\EndFor
\State $K:=\frac{1}{2}\left(K_0+\mathbf{1}_{m\times m}\right)$\Comment{$\SVM$ training}
\State Run the dual program~\eqref{eq:L2dualnobiasvectorized}, record solution as $\alpha$
\State\textbf{Return} $\alpha$
\end{algorithmic}
\end{algorithm}

\begin{algorithm}[H]
\caption{Support vector machine with quantum kernel estimation ($\SVM$-$\QKE$ testing)}\label{alg:qketesting}
\textbf{Input:} a new example $x\in \calX$, a training set $S=\{(x_i,y_i)\}_{i=1}^m$, training parameters $\alpha_1,\ldots,\alpha_m$\\
\textbf{Output:} $y\in \pmset{}$
\begin{algorithmic}[1]
\State $t:=0$
\For {$i=1\dots m$}\Comment{quantum kernel estimation}
    \State Apply $U^\dag(x)U(x_i)$ on input $\ket{0^n}$
    \State Measure the output probability of $0^n$ with $R$ shots, denoted as $p$
    \State $t:=t+\alpha_i y_i\cdot \frac{p+1}{2}$
\EndFor
\State\textbf{Return} $\sign(t)$
\end{algorithmic}
\end{algorithm}

We describe the full support vector machine algorithm with quantum kernel estimation ($\SVM$-$\QKE$), with training (Algorithm~\ref{alg:qketraining}) and testing (Algorithm~\ref{alg:qketesting}) phases. Here, $\mathbf{1}$ denotes the all-one matrix, and $R$ denotes the number of measurement shots for each kernel estimation circuit. 

The main differences between $\QKE$ and classical kernels are two-fold:
\begin{itemize}
    \item On the one hand, quantum feature maps are more expressive than classical feature maps. Therefore, a $\SVM$ trained with a quantum kernel may achieve better performance than with classical kernels.
    \item On the other hand, a fundamental feature in $\QKE$ is that we only have a noisy estimate of the quantum kernel entries in both training and testing, due to finite sampling error in experiment. More specifically, for each $K_0(x_i,x_j)$ as defined in Algorithm~\ref{alg:qketraining}, we have access to a noisy estimator $p$ with mean equals $K_0(x_i,x_j)$ and variance $\frac{1}{R}$.
\end{itemize}
Therefore, noise robustness is an important property for provable quantum advantage with $\QKE$. We will formally prove this in the next section.

\section{Efficient learnability with quantum kernel estimation}\label{app:qkelearnability}

\subsection{Quantum feature map}\label{app:qfeaturemap}
We now define our quantum feature map for learning the concept class $\mc C$ based on the discrete logarithm problem (recall Definition~\ref{def:conceptclass}). This family of states, whose construction is based on the discrete logarithm problem, was first introduced in Ref.~\cite{Aharonov2007adiabatic} to study the complexity of quantum state generation and statistical zero knowledge. 

For a prime $p$, let $n=\lceil\log_2 p\rceil$ be the number of bits needed to represent $\{0,1,\dots,p-1\}$. 
For $y\in \ints_p^*$, $k\in\{1,2,\dots,n-1\}$, define a polynomial-sized classical circuit family $\{C_{y,k}\}$ as follows:
\begin{nalign}
&C_{y,k}:\{0,1\}^k\to \{0,1\}^n,\\
&C_{y,k}(i)=y\cdot g^i\pmod p,\forall i\in\{0,1\}^k.
\end{nalign}
It's easy to see that $C_{y,k}$ is injective, i.e.,  for all $ i\neq j$, we have $C_{y,k}(i)\neq C_{y,k}(j)$. Furthermore, $C_{y,k}(i)$ can be computed using $\mc O(n)$ multiplications within $\ints_p^*$. We now show how to prepare a uniform superposition over the elements of $C_{y,k}$ on a quantum computer, which we refer to as the $n$-qubit feature state
\begin{nalign}
\ket{C_{y,k}}=\frac{1}{\sqrt{2^k}}\sum_{i\in\{0,1\}^k}\ket{y\cdot g^i}.
\end{nalign}
 First construct the reversible extension of $C_{y,k}$ as
\begin{nalign}
    &\tilde{C}_{y,k}:\{0,1\}^{2n}\to \{0,1\}^{2n},\\
    &\tilde{C}_{y,k}\ket{i}\ket{0^n}=\ket{i}\ket{C_{y,k}(i)},
\end{nalign}
where $\ket{i}$ uses the least significant bits of the $n$-bit register. Then, construct a quantum circuit~$U_y$ using the quantum algorithm for discrete log~\cite{Shor1997polynomial,Mosca2004exact,gidney2019factor} which uses $\Tilde{\mc O}(n^3)$ gates (we use $\Tilde{\mc O}(\cdot)$ to hide $\polylog$ factors), 
\begin{nalign}
    U_y\ket{C_{y,k}(i)}\ket{0}=\ket{i}\ket{C_{y,k}(i)}.
\end{nalign}
The overall procedure for preparing $\ket{C_{y,k}}$ is as follows, up to adding/discarding auxiliary qubits:
\begin{nalign}
    \ket{0^n}&\xrightarrow{H^{\otimes k}}\frac{1}{\sqrt{2^k}}\sum_{i\in\{0,1\}^k}\ket{i}\xrightarrow{\tilde{C}_{y,k}}\frac{1}{\sqrt{2^k}}\sum_{i\in\{0,1\}^k}\ket{i}\ket{C_{y,k}(i)}\xrightarrow{U_y^\dag}\frac{1}{\sqrt{2^k}}\sum_{i\in\{0,1\}^k}\ket{C_{y,k}(i)}.
\end{nalign}
\begin{definition}\label{def:featurestate}
Define the family of \emph{feature states} via the map $(y,k)\to\ket{C_{y,k}}$ that takes classical data $y\in \ints_p^*,k\in\{1,2,\dots,n-1\}$ and maps it to a $n$-qubit feature state 
\begin{nalign}
\ket{C_{y,k}}=\frac{1}{\sqrt{2^k}}\sum_{i\in\{0,1\}^k}\ket{y\cdot g^i}.
\end{nalign}
Such a procedure can be implemented in $\BQP$ using $\Tilde{\mc O}(n^3)$ gates.
\end{definition}

In Ref.~\cite{Aharonov2007adiabatic}, it was proven that constructing these feature states is as hard as solving discrete log, where they show that $\DLP_{1/6}$ can be reduced to estimating the inner product between $\ket{C_{y,k}}$ with different $y$ and $k$. In this work, our quantum kernel is constructed via the feature map with a fixed~$k$, which is chosen prior to running the quantum kernel estimation algorithm. More specifically, for different training samples $y,y'\in\ints_p^*$, the corresponding kernel entry is given by
\begin{equation}
    K_0(y,y')=\left|\braket{C_{y,k}}{C_{y',k}}\right|^2.
\end{equation}

We note that these feature states have a special structure: after taking the discrete log for each basis state, the feature state becomes the superposition of an interval. Therefore, computing the inner product between feature states is equivalent to computing the intersection between their corresponding intervals. We provide more details in the following definition.

\begin{definition}[Interval states]
For a fixed $g,p$, suppose $y= g^x$. The feature state can be written as $\ket{C_{y,k}}=\frac{1}{\sqrt{2^k}}\sum_{i=0}^{2^k-1}\ket{g^{x+i}}$. This can be understood as ``interval states" in the log space, where the exponent spans an interval $[x,\dots,x+2^k-1]$ of length $2^k$. As a consequence, since $y=g^x$ is a one-to-one mapping, computing the inner products between feature states is equivalent to computing intersection of the corresponding intervals.
\end{definition}

By definition, our kernel $K_0$ is constructed using interval states with a fixed length. In order for our quantum algorithm to solve a classically hard problem, a necessary condition is that the kernel entries $K_0(y,y')$ cannot be efficiently estimated up to additive error. Otherwise, the quantum kernel estimation procedure can be efficiently simulated classically. Next we show that estimating the kernel entries is as hard as solving the discrete log problem. Although this is implied by our main results (Theorem~\ref{thm:conceptclasshardness} and \ref{thm:qkelearnability}), here we give a direct proof which is a generalization of Ref.~\cite{Aharonov2007adiabatic}.
\begin{lemma}\label{lemma:kernelestimationhardness}
For an arbitrary (fixed) prime $p$ and generator $g$, if there exists a polynomial time algorithm such that, on input $y,y'\in \ints_p^*$, computes $K_0(y,y')$ up to 0.01 additive error, then there exists a polynomial time algorithm for $\DLP(p,g)$.
\end{lemma}
\begin{proof}
We show this lemma by using an algorithm that estimates the kernel entries well to solve $\DLP_{\frac{1}{16}}(p,g)$ which in turn (by Lemma~\ref{lemma:dlpreducetopromise}) implies an efficient algorithm for $\DLP(p,g)$. In the following we assume $k=n-3$, but the proof can be generalized to any $k=n-t\log n$ for some constant $t$.

Consider an input $y=g^x$ for the problem $\DLP_{\frac{1}{16}}(p,g)$, where we are promised that either $x\in[1,\frac{p-1}{16}]$ or $x\in[\frac{p-1}{2}+1,\frac{p-1}{2}+\frac{p-1}{16}]$. Let $y'=g^{(p+1)/2}$ and consider feature states $\ket{C_y}$ and~$\ket{C_{y'}}$. Then for the two cases,
\begin{enumerate}
    \item If $x\in[1,\frac{p-1}{16}]$, $\ket{C_y}$ corresponds to a subinterval of $[1,\frac{p-1}{2}]$ and therefore $K_0(y,y')=0$.
    \item If $x\in[\frac{p-1}{2}+1,\frac{p-1}{2}+\frac{p-1}{16}]$, the intersection of the corresponding intervals is at least $\frac{p}{16}$, so $K_0(y,y')\geq \frac{1}{16}$.
\end{enumerate}
Therefore, an algorithm that can approximate $K_0(y,y')$ within $0.01$ additive error can decide the promise problem $\DLP_{\frac{1}{16}}(p,g)$. Lemma~\ref{lemma:dlpreducetopromise} now shows the lemma statement.
\end{proof}

\subsection{Mapping to high dimensional Euclidean space}
Now we are ready to apply the feature map in our support vector machine algorithm using quantum kernel estimation. We recall the definition of $\mc C$ (Definition~\ref{def:conceptclass}) here for convenience: $\mc C=\{f_s\}_{s\in \ints_p^*}$, where
\begin{equation}
    f_s(x)=\begin{cases}+1,&\text{if }\log_g x\in[s,s+\frac{p-3}{2}],\\ -1, &\text{else.}\end{cases}
\end{equation}
Consider the mapping from $x\in\ints_p^*$ to the quantum feature states described in the previous~section (renamed here as $\ket{\phi(x)}$),
\begin{equation}
    x\to\ket{\phi(x)}=\frac{1}{\sqrt{2^k}}\sum_{i=0}^{2^k-1}\ket{x\cdot g^i},
\end{equation}
where $k=n-t\log n$ for some constant $t$ to be specified later (recall that $n=\lceil\log_2 p\rceil$). 
It was shown in Definition~\ref{def:featurestate} that $\ket{\phi(x)}$ can be prepared in $\BQP$. Let $\Delta=\frac{2^{k+1}}{p}=\mc O(n^{-t})$. Then the feature states span a $\frac{\Delta}{2}=\mc O(n^{-t})$ fraction of the elements in $\ints_p^*$.

Also define the \emph{halfspace state} $\ket{\phi_s}$ corresponding to every concept $c_s\in \Cc$ as follows
\begin{equation}
    \ket{\phi_s}=\frac{1}{\sqrt{(p-1)/2}}\sum_{i=0}^{(p-3)/2}\ket{g^{s+i}},\,\,\,\,\text{ for every }s\in \ints_p^*.
\end{equation}
Observe that the halfspace state spans a $\frac{1}{2}$-fraction of the full space $\ints_p^*$. The following property shows that $\ket{\phi_s}$ is a separating hyperplane in Hilbert space.
\begin{itemize}
    \item $|\braket{\phi_s}{\phi(x)}|^2=\Delta$, for $1-\Delta$ fraction of $x$ in $\{x:f_s(x)=+1\}$.
    \item $|\braket{\phi_s}{\phi(x)}|^2=0$, for $1-\Delta$ fraction of $x$ in $\{x:f_s(x)=-1\}$.
\end{itemize}

Notice that $\ket{\phi_s}$ has a \emph{large margin} property: it separates training samples with label $+1$ from those with label $-1$. The probability of having an outlier (a data point that lies inside the margin or on the wrong side of the hyperplane) is small, which equals $\Delta=1/\poly(n)$. Recall that the goal of an $\SVM$ algorithm is to find a hyperplane that maximizes the margin on the training set, and the above property shows that such a good hyperplane exists.

In general, learning a separating hyperplane that separates $+1/-1$ examples is called a \emph{halfspace learning} problem. Rigorously speaking, our data vectors are represented by quantum states with the Hilbert-Schmidt inner product. For simplicity, we now show that our learning problem is equivalent to learning a halfspace in $4^n$-dimensional Euclidean space. For that, we first express a Hermitian matrix $W\in \mathbb{C}^{2^n\times 2^n}$ uniquely in terms of the orthonormal Pauli basis as follows

\begin{equation}
    W=\frac{1}{\sqrt{2^n}}\sum_{p\in\{0,1,2,3\}^n}w_p\sigma_p,
\end{equation}
where $\sigma_p\in\{\textsf{I},\textsf{X},\textsf{Y},\textsf{Z}\}^{\otimes n}$ are $n$-qubit Pauli operators and $w_p=\frac{1}{\sqrt{2^n}}\Tr[\sigma_p W]\in\mathbb{R}$ are the \emph{Fourier coefficients}. We can use the $4^n$ dimensional \emph{Pauli vector} $w=(w_p)$ to represent $W$, as the Hilbert-Schmidt inner product $\langle W,W'\rangle=\langle w,w'\rangle$ is equivalent to the Euclidean inner product. Also note that a pure quantum state in Hilbert space corresponds to a unit Pauli vector in Euclidean space.

The large margin property can be recast in Euclidean space with the Pauli basis representation.
\begin{lemma}\label{lemma:largemargin}
    For any concept $f_s\in\mc C$, let $w_s$ be the Pauli vector of $\ketbra{\phi_s}$, $b=-\frac{\Delta}{2}$, and $\hat{x}$ be the Pauli vector of $\ketbra{\phi(x)}$. Then
    \begin{itemize}
        \item $\langle w_s,\hat{x}\rangle + b =\frac{\Delta}{2}$, for $1-\Delta$ fraction of $x$ in $\{x:f_s(x)=+1\}$,
        \item $\langle w_s,\hat{x}\rangle + b =-\frac{\Delta}{2}$, for $1-\Delta$ fraction of $x$ in $\{x:f_s(x)=-1\}$,
    \end{itemize}
    where $\langle\cdot,\cdot\rangle$ denotes Euclidean inner product.
\end{lemma}

Our $\SVM$ algorithm that uses the kernel $K_0(x,x')=\left|\braket{\phi(x)}{\phi(x')}\right|^2$ can be equivalently understood as using the kernel $K_0(x,x')=\langle \hat{x},\hat{x}'\rangle$ based on a feature map that maps $x\in\ints_p^*$ to $\hat{x}$, a $4^n$ dimensional vector in Euclidean space. Finally, recall that we only have access to a noisy estimate of $K_0(x,x')$ using quantum kernel estimation. The noisy estimator that we obtain from a quantum computer has mean $K_0(x,x')$ and variance $\frac{1}{R}$, where $R$ denotes the number of measurement shots for each kernel estimation circuit.

To summarize, here we have shown that the original problem of learning the concept class $\Cc$ can be mapped to a noisy halfspace learning problem~\footnote{Note that here the halfspace learning problem is defined in feature space; our original concept class is not a halfspace learning problem. Also, in our case the noise is defined in terms of additive error in the kernel, which is different from the well-studied question of \emph{learning halfspaces with noise} in computational learning theory.} in high dimensional Euclidean space with the following properties. For simplicity, below we do not specify the concept, since everything holds equivalently for each concept in $\Cc$. From now on our analysis is restricted to the high dimensional Euclidean space, and for notation simplicity we overload $x$ to represent the Pauli vector $\hat{x}$.

\vspace{0.2cm}
\noindent\textbf{Properties of noisy halfspace learning.}
\begin{enumerate}
    \item \emph{Data space}: $\calX\subseteq \mathbb{R}^{4^n}$ with unit length $\|x\|_2=1$ for every $x\in \calX$. Each $x$ is associated with a label $y\in \pmset{}$.
    \item \emph{Separability}: the data points lie outside a margin of $\frac{\Delta}{2}$ with high probability over the uniform distribution on $\calX$. 
    That is, there exists a hyperplane $(w,b)$ where $w\in \mathbb{R}^{4^n}$, $\|w\|_2=1$, and $b\in\mathbb{R}$, such that
    \begin{equation}
        \Pr_{x\sim \calX}\left[y(\langle w,x\rangle+b)\geq\frac{\Delta}{2}\right]=1-\Delta.
    \end{equation}
    \item \emph{Bounded distance}: all data points are close to the above hyperplane:
    \begin{equation}
        \left|y(\langle w,x\rangle+b)\right|\leq\frac{\Delta}{2},\,\,\,\,\text{ for every } x\in \calX.
    \end{equation}
    \item \emph{Noisy kernel}: instead of having the ideal kernel $K_0(x_i,x_j)=\langle x_i,x_j\rangle$, we have access to a noisy kernel $K_0'$, where $K_0'(x_i,x_j)=K_0(x_i,x_j)+e_{ij}$. Here, $e_{ij}$ are independent random variables satisfying
    \begin{itemize}
        \item $e_{ij}\in[-1,1]$
        \item $\E [e_{ij}]=0$
        \item $\Var[e_{ij}]\leq \frac{1}{R}$, where $R$ denotes the number of measurement shots.
    \end{itemize}
\end{enumerate}
These properties are simple corollaries of the definition of $\ket{\phi(x)}$, $\ket{\phi_s}$ and Lemma~\ref{lemma:largemargin}.

In order to further simplify our analysis, we perform an additional transform which allows us to remove the bias parameter $b$ without loss of generality. This will help us simplify our notations in later proofs. More specifically, we replace $x$ with $(x,1)/\sqrt{2}$ and $w$ with $(w,b)/\sqrt{w^2+b^2}$. This corresponds to replacing the original kernel $K_0$ with a new kernel $K=\frac{1}{2}\left(K_0+\mathbf{1}_{m\times m}\right)$ where $\mathbf{1}$ denotes the all-one matrix. These steps are explained in more detail in Section~\ref{section:tutorial}. 
The final form of our halfspace learning problem is given below.

\begin{lemma}\label{lemma:noisyhalfspacelearning}
    We have mapped the original problem of learning the concept class $\mc C$ into the following noisy halfspace learning problem. Below we do not specify the concept, as these properties hold for every concept in $\mc C$.
    \begin{enumerate}
    \item \emph{Data space:} $\calX\subseteq \mathbb{R}^{4^n+1}$ with unit length $\|x\|_2=1$, $\forall x\in \calX$. Each $x$ is associated with a label $y\in \pmset{}$.
    \item \emph{Separability:} the data points lie outside a $\mc O(\Delta)$ margin with high probability over the uniform distribution. That is, there exists a hyperplane $w$ where $w\in \mathbb{R}^{4^n+1}$, $\|w\|_2=1$, such that
    \begin{equation}
        \Pr_{x\sim \calX}\left[y\langle w,x\rangle\geq\frac{\Delta}{\sqrt{8+2\Delta^2}}\right]=1-\Delta.
    \end{equation}
    \item \emph{Bounded distance:} all data points are close to the above hyperplane:
    \begin{equation}
        \left|y\langle w,x\rangle\right|\leq\frac{\Delta}{\sqrt{8+2\Delta^2}},\,\,\,\,\text{ for every } x\in \calX.
    \end{equation}
    \item \emph{Noisy kernel:} let $K_{ij}=\frac{1}{2}\left(1+K_0(x_i,x_j)\right)$. We have access to a noisy kernel $K'$, where $K_{ij}'=K_{ij}+e_{ij}$. Here, $e_{ij}$ are independent random variables satisfying
    \begin{itemize}
        \item $e_{ij}\in[-1/2,1/2]$
        \item $\E [e_{ij}]=0$
        \item $\Var[e_{ij}]\leq \frac{1}{R}$, where $R$ denotes the number of measurement shots.
    \end{itemize}
\end{enumerate}
\end{lemma}

The hyperplane specified in the above lemma is particularly useful for our analysis later. We define its unnormalized version as the ``ground truth hyperplane" as follows.

\begin{definition}\label{def:groundtruth}
Consider the halfspace learning problem defined in Lemma~\ref{lemma:noisyhalfspacelearning}. Define $w^*\in\mathbb{R}^{4^n +1}$ as the (unnormalized) ground truth hyperplane as given in Lemma~\ref{lemma:noisyhalfspacelearning}, that satisfies the following~properties:
\begin{nalign}
    &\Pr_{x\sim \calX}\left[y\langle w^*,x\rangle\geq 1\right]=1-\Delta,\\
    &|y\langle w^*,x\rangle|\leq 1,\,\,\,\,\text{ for every } x\in \calX.
\end{nalign}
Note that the norm of $w^*$ is $\|w^*\|_2=\mc O(\Delta^{-1})$.
\end{definition}

\subsection{Generalization of the noisy classifier}
Next, we focus on the noisy halfspace learning problem as given by Lemma~\ref{lemma:noisyhalfspacelearning}. We show that the four properties established in Lemma~\ref{lemma:noisyhalfspacelearning} are sufficient for formally proving the efficient learnability of the concept class $\mc C$ using our quantum algorithm.

Consider the primal optimization problem in the support vector machine used by Algorithm~\ref{alg:qketraining}:
\begin{nalign}\label{eq:L2primal}
     \min_{w,\xi}\,\,\,\,\,\,\,\,&\frac{1}{2}\|w\|_2^2 + \frac{\lambda}{2}\sum_i\xi_i^2\\
     \text{s.t.}\,\,\,\,\,\,\,\,&y_i\langle x_i,w\rangle\geq 1-\xi_i\\
     &\xi_i\geq 0
\end{nalign}
with the dual form
\begin{nalign}\label{eq:L2dual}
     \max_{\alpha}\,\,\,\,\,\,\,\,&\sum_i\alpha_i-\frac{1}{2}\sum_{i,j}\alpha_i\alpha_j y_i y_j K_{ij}-\frac{1}{2\lambda}\sum_i\alpha_i^2\\
     \text{s.t.}\,\,\,\,\,\,\,\,&\alpha_i\geq0.
\end{nalign}
The above duality follows from the KKT conditions $w=\sum_i \alpha_i y_i x_i$ and $\alpha_i=\lambda\xi_i$. The kernel matrix $K$ is a positive semidefinite matrix. Let $Q$ be a matrix such that $Q_{ij}=y_iy_jK_{ij}$, which is also positive semidefinite. Then the dual program \eqref{eq:L2dual} is equivalent to the following convex quadratic program:
\begin{nalign}\label{eq:quadraticprogram}
     \min_{\alpha}\,\,\,\,\,\,\,\,&\frac{1}{2}\alpha^T\left(Q+\frac{1}{\lambda}\mathbb{I}\right)\alpha-1^T\alpha\\
     \text{s.t.}\,\,\,\,\,\,\,\,&\alpha\geq0.
\end{nalign}

Recall that we have small additive perturbations in $K$, which in turn gives additive perturbations in $Q$. One useful property of the dual program \eqref{eq:quadraticprogram} is that it is robust to perturbations in $Q$. More specifically, we use the following lemma from standard perturbation analysis.

\begin{lemma}[{\cite[Theorem 2.1]{Daniel1973}}]\label{lemma:quadraticprogramrubostness}
Let $x_0$ be the solution to the quadratic program
\begin{nalign}\label{eq:robustquadraticprogram}
     \min_{x}\,\,\,\,\,\,\,\,&\frac{1}{2}x^T K x-c^T x\\
     \text{s.t.}\,\,\,\,\,\,\,\,&Gx\leq g\\
     &Dx=d,
\end{nalign}
where $K$ is positive definite with smallest eigenvalue $\lambda>0$. Let $K'$ be a positive definite matrix such that $\|K'-K\|_F\leq \varepsilon<\lambda$. Let $x_0'$ be the solution to \eqref{eq:robustquadraticprogram} with $K$ replaced by $K'$. Then 
\begin{equation}
    \|x_0'-x_0\|_2\leq \frac{\varepsilon}{\lambda-\varepsilon}\|x_0\|_2.
\end{equation}
\end{lemma}

Before going into the analysis of robustness against noise, notice that the bound in Lemma~\ref{lemma:quadraticprogramrubostness} is multiplicative. Therefore, it is useful to establish an upper bound on $\|\alpha\|_2$, the norm of the solution to the noiseless quadratic program~\eqref{eq:L2dual}. Recall from the KKT conditions that $\alpha_i=\lambda\xi_i$, where the dual variables are directly related to the slack variables in the primal program~\eqref{eq:L2primal}. The following lemma establishes a useful property for $\xi_i^*$ for the ground truth hyperplane.

\begin{lemma}\label{lemma:boundxi}
    For the ground truth hyperplane $w^*$ as defined in Definition~\ref{def:groundtruth}, the corresponding slack variables $\xi_i^*$ in the primal program~\eqref{eq:L2primal} satisfy 
    \begin{equation}
        \E\left[\|\xi^*\|_2^2\right]\leq \mc O(m\Delta),
    \end{equation}
    where the expectation is taken over the training set.
\end{lemma}
\begin{proof}
We can write the slack variables as
\begin{equation}
    \xi_i^*=\max\{1-y_i\langle x_i,w^*\rangle,0\}.
\end{equation}
By the properties given in Definition~\ref{def:groundtruth}, we have
\begin{nalign}
    &\Pr\left[\xi_i^*=0\right]=1-\Delta,\\
    &\xi_i^*\leq 2.
\end{nalign}
Therefore,
\begin{equation}
    \E\left[\|\xi^*\|_2^2\right]= m\E\left[\xi_1^{*2}\right]\leq 4m\Delta.
\end{equation}
\end{proof}

Now we are ready to bound the norm of the dual variables $\alpha_i$. Intuitively, we can do so because $\|\xi\|_2^2$ is part of the training loss in the primal program~\eqref{eq:L2primal}. The loss of the solution returned by the program can only be smaller than the loss of the ground truth hyperplane, which is guaranteed to be small. 

\begin{lemma}\label{lemma:boundalpha}
    Let $\alpha_0$ be the solution returned by the dual program~\eqref{eq:L2dual}. We have
    \begin{equation}
        \E\left[\|\alpha_0\|_2^2\right]=\mc O\left(\frac{1}{\Delta^2}+m\Delta\right),
    \end{equation}
    where the expectation is over the training set.
\end{lemma}
\begin{proof}
Let $w_0=\sum_i \alpha_{0i}y_i x_i$ be the hyperplane which corresponds to $\alpha_0$, and let $\xi_0$ be the corresponding slack variable. Then
\begin{nalign}
    \|\alpha_0\|_2^2&=\lambda^2\|\xi_0\|_2^2\leq \lambda\left(\|w_0\|_2^2+\lambda\|\xi_0\|_2^2\right)\leq \lambda\left(\|w^*\|_2^2+\lambda\|\xi^*\|_2^2\right).
\end{nalign}
Here, the first line follows from the KKT condition $\alpha_i=\lambda\xi_i$, and the third line is because $w_0$ is the optimal solution to \eqref{eq:L2primal}. Therefore by Lemma~\ref{lemma:boundxi}, $\E\left[\|\alpha_0\|_2^2\right]=\mc O\left(\frac{1}{\Delta^2}+m\Delta\right)$.
\end{proof}

\begin{remark}\label{rmk:boundalpha}
Recall that we have the freedom to choose $\Delta=\mc O(n^{-t})$ for any constant $t$. Suppose we have polynomially many training samples $m\approx n^c$, and let $t=c/3$. Then the above bound gives $\E\left[\|\alpha_0\|_2^2\right]=\mc O\left(m^{2/3}\right)$.
\end{remark}

Having established the above lemmas, now we are ready to prove our key result for noise robustness (Lemma~\ref{lemma:noiserobustness}). Let $Q'$ be the noisy kernel measured by a quantum computer, and let $\lambda\in (0,1)$ be a constant. Here we briefly recall the steps in Algorithm~\ref{alg:qketraining} and \ref{alg:qketesting}. The classifier is constructed in two steps. First, use a classical computer to run the dual program \eqref{eq:quadraticprogram} with $Q$ replaced by the experimental estimate $Q'$, and let $\alpha'$ be the solution returned by the program. 
Second, given a new data sample $x$, use a quantum computer to obtain noisy estimates $K'(x,x_i)$ for all $i$, and output
\begin{equation}
    y_{\text{pred}}=\sign\left(\sum_i\alpha_i' y_i K'(x,x_i)\right).
\end{equation}

Let $h(x)=\sum_i\alpha_i y_i K(x,x_i)$ and $h'(x)=\sum_i\alpha_i' y_i K'(x,x_i)$, which corresponds to the value of the noiseless/noisy classifier before taking the sign. We will prove the following result which establishes the noise robustness of $h$.

\begin{lemma}[Noise robustness]\label{lemma:noiserobustness}
Suppose we take $R=\mc O(m^4)$ measurement shots for each quantum kernel estimation circuit. Then, with probability at least 0.99 (over the choice of random training samples and measurement noise), for every $x\in \calX$ we have
\begin{equation}
    \left|h(x)-h'(x)\right|\leq 0.01.
\end{equation}
\end{lemma}
\begin{proof}
Consider the (noisy) quadratic program \eqref{eq:L2dual}. The Frobenius norm is given by
\begin{equation}
    \|Q'-Q\|_F^2=2\sum_{i<j}e_{ij}^2,
\end{equation}
where $e_{ij}$ are independent random variables satisfying $\E[e_{ij}]=0$ and $\E\left[e_{ij}^2\right]\leq \frac{1}{R}$. Therefore $\E\left[\|Q'-Q\|_F^2\right]\leq\mc O\left(\frac{m^2}{R}\right)$. Now we invoke Lemma~\ref{lemma:quadraticprogramrubostness} and Lemma~\ref{lemma:boundalpha} (see Remark~\ref{rmk:boundalpha}). Using  Markov's inequality: with probability at least 0.999 (over the choice of training samples and measurement noise), we have that \begin{nalign}
    &\|Q'-Q\|_F^2\leq\mc O\left(\frac{m^2}{R}\right), \text{ and }\hspace{2mm} \|\alpha\|_2^2=\mc O\left(m^{2/3}\right).
\end{nalign}
 Let $\delta_i=\alpha_i'-\alpha_i$. Since  $\lambda_{\min}\left(Q+\frac{1}{\lambda}\mathbb{I}\right)\geq\frac{1}{\lambda}$ is lower bounded by a constant, Lemma~\ref{lemma:quadraticprogramrubostness} gives
\begin{equation}
    \|\delta\|_2\leq \mc O\left(\frac{m^{4/3}}{\sqrt{R}}\right).
\end{equation}
Then, let $\nu_i=K'(x,x_i)-K(x,x_i)$ for $i=1,\dots,m$. Similarly, 
$\nu_i$ are independent random variables satisfying $\E[\nu_i]=0$ and $\E\left[\nu_i^2\right]\leq \frac{1}{R}$. By Markov's inequality, we have $\|\nu\|_2\leq\mc O\left(\frac{\sqrt{m}}{\sqrt{R}}\right)$ with probability at least 0.999. Overall for any $x\in \calX$, the error bound gives
\begin{nalign}
    \left|h(x)-h'(x)\right|&=\left|\sum_i(\alpha_i+\delta_i) y_i (K(x,x_i)+\nu_i)-\sum_i\alpha_i y_i K(x,x_i)\right|\\
    &\leq \sum_i\left|\alpha_i\nu_i+\delta_i K(x,x_i)+\delta_i\nu_i\right|\\
    &\leq \|\alpha\|_2\cdot\|\nu\|_2+\sqrt{m}\|\delta\|_2+\|\delta\|_2\cdot\|\nu\|_2\\
    &\leq\mc O\left(\frac{m^{11/6}}{\sqrt{R}}\right),
\end{nalign}
where the third line uses Cauchy–Schwarz inequality. Therefore, $R=\mc O(m^4)$ measurement shots is sufficient for achieving $\left|h(x)-h'(x)\right|\leq 0.01$, and by a simple union bound this holds with probability at least 0.99.
\end{proof}

Having established noise robustness, it remains to prove a generalization error bound for the noisy classifier: if the classifier has small training error/loss, it should also have small test error,~which is referred to as \emph{generalization error} in learning theory. The main idea is a two-step~argument:
\begin{enumerate}
    \item The noiseless classifier $y=\sign\left(h(x)\right)$ (we have defined $h(x)=\sum_i\alpha_i y_i K(x,x_i)=\langle w,x\rangle$) has small generalization error, which follows from standard generalization bounds for soft margin classifiers.
    \item We have established that the noisy classifier is close to the noiseless classifier. Therefore, the noisy classifier should also have small generalization error.
\end{enumerate}
For the first step, we refer to standard results on the generalization of soft margin classifiers (see, for example~\cite{anthony_bartlett_2000,bartlett1998prediction,Shawe-Taylor1998structural,BST1999,Shawe-Taylor2002generalization,gronlund2020near}). Recall that a hyperplane $w$ correctly classifies a data point $(x,y)$ if and only if $y\langle w,x\rangle>0$. Therefore for a specific concept $f\in\mc C$, the test accuracy of $f^*(x)=\sign\left(\langle w,x\rangle\right)$ is given by 
\begin{equation}
    \mathrm{acc}_{f}(f^*)=\Pr_{x\sim \calX}\left[f^*(x)=f(x)\right]=1-\Pr_{x\sim \calX}\left[y\langle w,x\rangle<0\right],
\end{equation}
where we have used $y=f(x)$. Our results will be given in the form of an upper bound on $\Pr_{x\sim \calX}\left[y\langle w,x\rangle<0\right]$. The following result gives a generalization bound that coincides with our $L2$ training loss up to $\polylog$ factors, as indicated by the $\Tilde{\mc O}$ notation, and therefore is directly applicable to the noiseless classifier.

\begin{lemma}[{\cite[Theorem~VII.11]{Shawe-Taylor2002generalization}}]
For any hyperplane $w$ satisfying the constraints of the primal program~\eqref{eq:L2primal}, with probability $1-\delta$ over randomly drawn training set $S$ of size $m$, the generalization error is bounded by
\begin{equation}\label{eq:L2softmargingeneralization}
    \Pr_{x\sim \calX}[y\langle w,x\rangle<0]\leq \frac{1}{m}\Tilde{\mc O}\left(\|w\|_2^2+\|\xi\|_2^2+\log\frac{1}{\delta}\right).
\end{equation}
\end{lemma}

However, although this result establishes step 1, it cannot be directly applied to step 2: our noise robustness result, which states that $h'(x)$ is close to $h(x)$, does not guarantee that $\sign(h'(x))$ agrees well with $\sign(h(x))$. The above lemma implies that $h(x)$ is on the correct side of the origin with high probability, but it could still be very close to the origin, which may lead to a bad noisy classifier $\sign(h'(x))$.

A simple solution to this problem is to show a stronger generalization bound, which in addition to $h(x)$ being correct, also shows that $h(x)$ is bounded away from the origin. We indeed prove such a result, by combining ideas from the aforementioned references. Notice that the only difference between the following lemma and the previous lemma is that we replaced 0 with 0.1.

\begin{lemma}\label{lemma:margingeneralization}
For any hyperplane $w$ satisfying the constraints of the primal program~\eqref{eq:L2primal}, with probability $1-\delta$ over randomly drawn training set $S$ of size $m$, the generalization error is bounded~by
\begin{equation}\label{eq:strongL2softmargingeneralization}
    \Pr_{x\sim \calX}[y\langle w,x\rangle<0.1]\leq \frac{1}{m}\Tilde{\mc O}\left(\|w\|_2^2+\|\xi\|_2^2+\log\frac{1}{\delta}\right).
\end{equation}
\end{lemma}
The proof is presented in Section~\ref{section:generalizationbound}. Combining Lemma~\ref{lemma:noiserobustness} with Lemma~\ref{lemma:margingeneralization}, we arrive at our main theorem for the learnability of $\mc C$ with our quantum algorithm.

\begin{theorem}\label{thm:qkelearnability}
For any concept $f_s\in \Cc$, the $\SVM$-$\QKE$ algorithm returns a classifier with test accuracy at least 0.99 in polynomial time, with probability at least $2/3$ over the choice of random training samples and over noise.
\end{theorem}
\begin{proof}
Below we do not specify the concept, as the proof works equivalently for every concept $f_s\in \Cc$. Let $w^*$ be the ground truth hyperplane as in Definition~\ref{def:groundtruth}. Note that $\|w^*\|_2=\mc O(\Delta^{-1})$. Using Lemma~\ref{lemma:boundxi}, we have that with probability at least 0.99 over the choice of training samples, the $L2$ training loss of $w^*$ satisfies
\begin{equation}
    \mathrm{Loss}(w^*):=\frac{1}{2}\|w^*\|_2^2+\frac{\lambda}{2}\|\xi^*\|_2^2\leq \mc O\left(\frac{1}{\Delta^2}+m\Delta\right).
\end{equation}
Let $w_0$ be the optimal solution of the primal program \eqref{eq:L2primal}, and let $h(x)=\langle w_0,x\rangle$. Let $h'$ be the noisy classifier obtained by the dual program~\eqref{eq:L2dual}. By Lemma~\ref{lemma:noiserobustness}, for any $x\in \calX$ we have
\begin{equation}
    |y h'(x)-y h(x)|\leq 0.01,
\end{equation}
with probability at least 0.99 over the choice of training samples and noise. Therefore, by a simple union bound, with probability at least $2/3$, the test error of the noisy classifier is upper bounded~by
\begin{nalign}
    \Pr_{x\sim \calX}[y h'(x)<0]&\leq \Pr_{x\sim \calX}[y h'(x)<0.09]\\
    &\leq \Pr_{x\sim \calX}[y\langle w_0,x\rangle<0.1]\\
    &\leq \frac{1}{m}\Tilde{\mc O}(\mathrm{Loss}(w_0))\\
    &\leq \frac{1}{m}\Tilde{\mc O}(\mathrm{Loss}(w^*))\leq \Tilde{\mc O}\left(\frac{1}{m\Delta^2}+\Delta\right),
\end{nalign}
where in the third line we use Lemma~\ref{lemma:margingeneralization}, and the fourth line is because $w_0$ is the optimal solution to \eqref{eq:L2primal}. Finally, notice that the above bound holds for arbitrary $\Delta=\mc O\left(n^{-t}\right)$ for constant $t$. In order to optimize this bound, we can choose $t=c/3$ for $m=n^c$ (also see Remark~\ref{rmk:boundalpha}). This gives the final~bound
\begin{equation}
    \Pr_{x\sim \calX}[y h'(x)<0]\leq\Tilde{\mc O}\left(m^{-1/3}\right).
\end{equation}
Therefore, polynomially many training samples are sufficient for learning the concept class $\mc C$ with high accuracy.
\end{proof}

\section{Generalization bound for soft margin SVM}\label{section:generalizationbound}
In this section we prove Lemma~\ref{lemma:margingeneralization}, a generalization bound for the L2 soft margin $\SVM$ in Eq.~\eqref{eq:L2primal} (restated below for convenience).
\begin{nalign}\label{eq:L2primalrestate}
     \min_{w,\xi}\,\,\,\,\,\,\,\,&\frac{1}{2}\|w\|_2^2 + \frac{\lambda}{2}\sum_i\xi_i^2\\
     \text{s.t.}\,\,\,\,\,\,\,\,&y_i\langle x_i,w\rangle\geq 1-\xi_i\\
     &\xi_i\geq 0.
\end{nalign}
Let $w$ be an unnormalized feasible solution to Eq.~\eqref{eq:L2primalrestate}. The first step is to use a trick developed by \cite{Shawe-Taylor2002generalization}, that converts the soft margin problem to a hard margin problem by mapping to a larger space. Let $S=\{(x_1,y_1),\dots,(x_m,y_m)\}$ be the training set. Consider the mapping
\begin{nalign}
    &x\mapsto\tilde{x}=\left(x,\delta_x\right)\\
    &w\mapsto\tilde{w}=\left(w,\sum_{(x,y)\in S}y\delta_x \cdot  \xi(x,y,w)\right) 
\end{nalign}
where $\delta_x:\calX\to \{0,1\}$ is a function defined as $\delta_x(x')=1$ if and only if $x'=x$ 
and $\xi(x,y,w)=\max\{0,1-y\langle w,x\rangle\}$ are the slack variables used in Eq.~\eqref{eq:L2primalrestate}. Denote the enlarged space by $L_X$ and $\|\cdot\|$ its induced norm. The following useful properties hold for this transform:
\begin{enumerate}
    \item If $(x,y)\in S$, $y\langle \tilde{w},\tilde{x}\rangle\geq 1$.
    \item If $(x,y)\notin S$, $\langle \tilde{w},\tilde{x}\rangle=\langle w,x\rangle$.
    \item $\|\tilde{w}\|^2=\|w\|_2^2+\|\xi\|_2^2$.
    \item $\|\tilde{x}\|^2=2$.
\end{enumerate}

In the following, we can assume that the training data does not appear when testing the classifier, which is the case with high probability. By property 2, to bound the generalization performance of the hyperplane $w$, we only need to bound the generalization performance of $\tilde{w}$ in the enlarged space, which corresponds to a hard margin problem.

For the generalization error of hard margin classifiers, it is well-known that the generalization error bound is captured by the VC-dimension which characterizes the complexity of the classifier family.
Intuitively, a hard margin classifier corresponds to a ``thick" hyperplane in the data space, which reduces its complexity compared with margin-less hyperplanes. The relevant complexity measure in our results is the so-called \emph{fat-shattering dimension} which we define~now.

\begin{definition}[{\cite[Definition~III.4]{Shawe-Taylor2002generalization}}]
Let $\mc F\subseteq \{f:\calX\to\mathbb{R}\}$ be a set of real-valued functions. We say that a set of points $\hat{X}\subseteq X$ is $\gamma$-shattered by $\mc F$, if there are real
numbers $r_x$ indexed by $x\in \hat{X}$ such that for all binary vectors $b_x$ indexed by $x\in \hat{X}$, there is a function $f_b\in\mc F$ satisfying
\begin{equation}
    f_b(x)\begin{cases}\geq r_x+\gamma, & \text{if }b_x=1,\\\leq r_x-\gamma, & \text{otherwise}.\end{cases}
\end{equation}
The $\gamma$-fat-shattering dimension of ${\mc F}$, denoted $\fat_{\mc F}(\gamma)$, is the size of the largest set $\hat{X}$ that is $\gamma$-shattered by $\mc F$, if this is finite or infinity otherwise.
\end{definition}

Let $\calH$ be a set of linear functions that map from $L_X$ to $\mathbb{R}$, such that their norm equals to~$\|\tilde{w}\|$. Since the data vectors $\tilde{x}$ have bounded norm, the fat-shattering dimension of $\calH$ was shown~\cite{Shawe-Taylor2002generalization} to be bounded by
\begin{equation}
    \fat_{\calH}(\gamma)\leq \mc O\left(\frac{\|\tilde w\|^2}{\gamma^2}\right).
\end{equation}
Next we invoke the following lemma from Ref.~\cite{anthony_bartlett_2000} (also see~\cite{bartlett1998prediction}) which uses fat-shattering dimension to understand generalization bounds for learning real-valued concept classes. 
\begin{lemma}[{\cite[Corollary 3.3]{anthony_bartlett_2000}}]\label{lemma:fatshatteringgeneralization}
Let $\Cc,\calH$ be sets of functions that map from a set $\calX$ to $[0,1]$. Then for all $\eta,\gamma,\delta\in(0,1)$, for every $f\in \Cc$ and for every probability measure $\mc D$ on $\calX$, with probability at least $1-\delta$ (over the choice of $S=\{x_1,\dots,x_m\}$ where $x_i\sim\mc D$), if $h\in \calH$ and $|h(x_i)-f(x_i)|\leq \eta$ for $1\leq i\leq m$, then
\begin{equation}
    \Pr_{x\sim\mc D}\big[|h(x)-f(x)|\geq\eta+\gamma\big]\leq\frac{1}{m}\cdot \Tilde{\mc O}\left(\fat_\calH\left(\frac{\gamma}{8}\right)+\log\frac{1}{\delta}\right).
\end{equation}
\end{lemma}

To apply this lemma, consider the function $h(\tilde{x})=\frac{1}{\|\tilde{w}\|}\langle \tilde{w},\tilde{x}\rangle$. Let $\gamma_0=\frac{1}{\|\tilde{w}\|}$ and $\gamma=0.9\gamma_0$. Let $y=f(\tilde{x})\in\{-1,1\}$ be any labeling rule, and $S=\{(\tilde{x}_1,y_1),\dots,(\tilde{x}_m,y_m)\}$ be a training set. By the properties of the mapping, for all $\tilde{x}\in S$ we have $yh(\tilde{x})\geq \gamma_0$, which means that
\begin{equation}
    |h(\tilde{x})-f(\tilde{x})|=|yh(\tilde{x})-1|\leq 1-\gamma_0.
\end{equation}
Applying Lemma~\ref{lemma:fatshatteringgeneralization}, we have with probability at least $1-\delta$,
\begin{equation}
    \Pr\big[|h(\tilde{x})-f(\tilde{x})|\geq 1-0.1\gamma_0\big]\leq\frac{1}{m}\Tilde{\mc O}\left(\|\tilde{w}\|^2+\log\frac{1}{\delta}\right).
\end{equation}
Finally, note that $yh(\tilde{x})\leq 0.1\gamma_0$ implies that $|h(\tilde{x})-f(\tilde{x})|\geq 1-0.1\gamma_0$, therefore
\begin{equation}
    \Pr\left[yh(\tilde{x})\leq 0.1\gamma_0\right]\leq\frac{1}{m}\Tilde{\mc O}\left(\|\tilde{w}\|^2+\log\frac{1}{\delta}\right).
\end{equation}
This concludes the proof of Lemma~\ref{lemma:margingeneralization}, as

\begin{nalign}
    \Pr_{x\sim \calX}[y\langle w,x\rangle<0.1]&=\Pr_{x\sim \calX}[y\langle \tilde{w},\tilde{x}\rangle<0.1]\\
    &=\Pr_{x\sim \calX}[yh(\tilde{x})<0.1\gamma_0]\\
    &\leq\frac{1}{m}\Tilde{\mc O}\left(\|\tilde{w}\|^2+\log\frac{1}{\delta}\right)= \frac{1}{m}\Tilde{\mc O}\left(\|w\|_2^2+\|\xi\|_2^2+\log\frac{1}{\delta}\right).
\end{nalign}
\end{document}